%% file: main_arxiv.tex
\documentclass[letterpaper,10pt]{article}
\usepackage{geometry}
\newgeometry{vmargin={40mm}, hmargin={25mm,25mm}}
\input{macros.sty}
\usepackage[affil-it]{authblk}

\title{Non-linear Triple Changes Estimator for Targeted Policies}
\author[1]{Sina Akbari\thanks{Corresponding author: Sina Akbari, sina.akbari@epfl.ch}}
\author[1]{Negar Kiyavash}

\affil[1]{EPFL, Switzerland}
\date{}

\usepackage{hyperref}
\usepackage{natbib}

\theoremstyle{plain}
\newtheorem{theorem}{Theorem}[section]
\newtheorem{proposition}[theorem]{Proposition}
\newtheorem{lemma}[theorem]{Lemma}

\theoremstyle{definition}

\newtheorem{assumption}{}
\theoremstyle{remark}
\newtheorem{remark}[theorem]{Remark}

\def\icmltitlerunning#1{\gdef\@icmltitlerunning{#1}}
\icmltitlerunning{Triple Changes Estimator}

\begin{document}
\maketitle







\begin{abstract}
    The renowned difference-in-differences (DiD) estimator relies on the assumption of `parallel trends,' which does not hold in many practical applications.
    To address this issue, the econometrics literature has turned to the triple difference estimator. 
    Both DiD and triple difference are limited to assessing average effects exclusively.
    An alternative avenue is offered by the changes-in-changes (CiC) estimator, which provides an estimate of the entire counterfactual distribution at the cost of relying on (stronger) distributional assumptions.
    In this work, we extend the triple difference estimator to accommodate the CiC framework, presenting the `triple changes estimator' and its identification assumptions, thereby expanding the scope of the CiC paradigm.
    Subsequently, we empirically evaluate the proposed framework and apply it to a study examining the impact of Medicaid expansion on children's preventive care.
\end{abstract}

\section{Introduction}
In the domains of econometrics and quantitative social sciences, the \emph{difference-in-differences} (DiD) estimator and its extensions have emerged as indispensable tools for estimating causal effects in observational studies.
With its roots traced back to the work of \citet{snow1855mode} and further popularized by seminal works such as \citet{ashenfelter1984using} and \citet{card1993minimum},
DiD has provided a reliable method for evaluating the effect of policy interventions over time.
In particular, DiD stands out in studies where inference based on controlling for confounders or using instrumental variables is deemed unsuitable, and where pre-treatment information is available.
Researchers have actively pursued various extensions of the DiD estimator, aimed at enhancing its applicability and robustness \cite{athey2006identification, sofer2016negative, callaway2021difference, roth2023parallel}.

The DiD estimator hinges on the assumption of parallel trends, which states that the (average) outcome of both the control and treatment groups share exactly the same evolution trend in the absence of the treatment.
Mathematically speaking, this assumption translates to
\begin{equation}\label{eq:common-trends}
    \E{\Y(t_1)-\Y(t_0)\cond D=1} = \E{\Y(t_1)-\Y(t_0)\cond D=0},
\end{equation}
where $\Y(t)$ denotes the potential outcome associated with the absence of the treatment at time $t$, and $D$ denotes the assigned treatment.
In particular, $D=1$ and $D=0$ represent the treatment and control groups, respectively.
The validity of this assumption can be challenged, and parallel trends might get violated due to unobserved time-varying factors or dynamic changes in the study context.
This violation may result in biased estimates.
As a response to potential deviations from parallel trends, researchers have turned to more flexible models, such as the \emph{triple difference} framework \cite{gruber1994incidence}.
The triple difference estimator can be formulated as the difference between two DiD estimators.
Intuitively, the difference of two DiD estimators is unbiased, provided both estimators have the same bias.
Indeed, sole purpose of of subtracting the second DiD estimator is to  debias the first one.
Despite the prevalent use of the triple difference estimator, especially over the last two decades \cite{raifman2018association, sakurai2020relationship, han2016effect, chen2020triple, tai2001racial}, only recently a formal presentation of the framework and its identification assumptions was provided \cite{olden2022triple}.

\import{./figures/}{fig0.tex}

The triple difference framework guarantees only the identification of the \emph{average treatment effect} (on the treated).
Another issue that arises is that the triple difference estimator is biased if the outcomes in the pre- and post-treatment or among control and treatment groups are measured on a different scale.
The \emph{changes-in-changes} (CiC) framework \citep{athey2006identification}, on the other hand, is scale-invariant and yields the identification of the counterfactual outcome probability distribution.
To do so, CiC framework requires additional assumptions beyond Eq.~\eqref{eq:common-trends}.
It postulates that there exists a unique monotone mapping $T_d$ which maps the probability measure over $\Y(t_0)$ to that over $\Y(t_1)$ within group $D=d$ for $d\in\{0,1\}$ (see Figure \ref{fig:one}.)
Moreover, these two mappings are assumed to be identical.
That is, $T_1(y) = T_0(y)$ for every $y$, or since the mappings are bijective, 
\begin{equation}\label{eq:tid}
    T_1\circ T_0^{-1} = \textrm{Id},
\end{equation}
where $\textrm{Id}$ represents the identity map.
Eq.~\eqref{eq:tid} states that there is \emph{no drift} in the evolution trend across groups.
~Note that Eq.~\eqref{eq:tid}, was not explicitly stated in \citet{athey2006identification} as an assumption but can be derived as a consequence of the following four assumptions:
\begin{assumption}[Model assumption]\label{as:model}
    The potential outcomes $\Y(t)$ can be modelled using a production function $h$ of a latent variable $U$, which models the individual characteristics:
    \begin{equation*}
        \forall t\in\{0,1\}:\quad \Y(t) = h(U; t).
    \end{equation*}
\end{assumption}
In particular, \ref{as:model} posits that $h$ does not depend on the group assignment ($D$).
\begin{assumption}[Strict monotonicity]\label{as:monotone}
    The function $h(\,\cdot\,; t)$ is strictly increasing in $U$ for every $t\in\{t_0,t_1\}$.
\end{assumption}
\begin{assumption}[Time invariance]\label{as:invariance}
    Within every subgroup, the distribution of the latent variable $U$ does not change over time.
\end{assumption}
\begin{assumption}[Latent support overlap]\label{as:sup}
    The support of the latent variable $U$ in the treatment group is a subset of its support in the control group.
\end{assumption}

In the case of a one-dimensional outcome, mappings $T_d$ can be expressed as
\begin{equation}\label{eq:td}
    T_d= F^{-1}_{\Y(t_1)\cond D=d}\circ F_{\Y(t_0)\cond D=d}
\end{equation} 
where $F_{\Y(t)\cond D=d}(\cdot)$ represents the cumulative density function of $\Y(t)$ in group $D=d$.
Under regularity conditions, $T_d$ 
is the unique monotone map that pushes forward the probability measure over $\Y(t_0)$ to that over $\Y(t_1)$ in group $D=d$ \citep{villani2009optimal, santambrogio2015optimal}.
Combining Equations \eqref{eq:tid} and \eqref{eq:td}, it is straightforward to identify the counterfactual distribution of $\Y(t_1)$ in the treatment group, $F_{\Y(t_1)\cond D=1}(y)$.
Specifically, 
\[  
\begin{split}
    F_{\Y(t_1)}&_{\cond D=1}(y)=
    F_{\Y(t_0)\cond D=1}\circ
    F^{-1}_{\Y(t_0)\cond D=0}\circ
    F_{\Y(t_1)\cond D=0}(y)
    ,
\end{split}
    \]
which matches Eq.~(9) in the original work of \citet{athey2006identification}.

The no-drift assumption specified in Eq.~\eqref{eq:tid} can be challenged in practice, especially in scenarios where the treatment or exposure is directed toward a specific sub-population. 
This situation arises in studies on the impact of targeted interventions, such as a criminal justice initiative ($D$), on recidivism rates, 
where this intervention is exclusively administered to individuals with specific criminal histories or risk profiles. Naturally, it is expected that the time evolution of counterfactual recidivism rates will exhibit significant disparities between the control and treated groups.
Similar challenges arise when the eligibility criteria is narrow within the context of social programs such as welfare or housing assistance, which target specific demographic groups.

Recognizing the strengths and weaknesses of both the triple difference and CiC frameworks, we propose a novel estimator that combines the best of both worlds.
Our proposed `triple changes' estimator aims to overcome the limitations of DiD and the triple difference estimator by leveraging the scale invariance and strong identification results of CiC, while introducing a more flexible mapping assumption that allows us to relax \eqref{eq:tid}.
We briefly present the contributions of this work.
\begin{itemize}[leftmargin=*]
    \item We formally present and analyze the triple changes estimator as an extension to CiC, and provide the necessary assumptions for its point identification first in the scalar case.
    We then discuss how to generalize our results to high-dimensional outcomes by harnessing theory of optimal transport. 
    \item We provide several partial identification results under relaxed versions of our proposed point identifiability assumptions.
    Further, we show the validity of analogous results for the classic CiC framework as a special case of our derivations.
    \item We introduce a finite-sample estimator for the average treatment effect on the treated within our framework and analyze its asymptotic behaviour.
    \item We conduct an empirical evaluation of our estimator on both synthetic and real datasets.
\end{itemize}

This paper is organized as follows.
Section \ref{sec:model} reviews the necessary background and the setup of the study.
In Section \ref{sec:one}, the identification of our estimand of interest is studied for a scalar outcome.
Section \ref{sec:inf} provides an estimator for the latter and studies its asymptotic properties.
In Section \ref{sec:ot}, we extend our work to high-dimensional outcomes using theory of optimal transport.
Numerical evaluations are presented in Section \ref{sec:exp}.

\subsection{Causal model}\label{sec:model}
We consider a study where we have access to data from two sources, e.g., two states of the united states, or two cities, or any two separate populations.
These two sources of data will be denoted by $S=s_0$ and $S=s_1$ throughout.
We assume that a treatment (e.g., a health-care policy) is administered in one state, without loss of generality in $S=s_1$, and not in the other.
In both states, the individuals are partitioned into two cohorts, namely, $D=d_1$ and $D=d_0$, signifying the individuals that are eligible and not eligible for receiving the treatment, respectively\footnote{In a classic controlled trial, these would correspond to the treatment and control groups, respectively.}. 
The outcome is measured in two time points, namely $t_0<t_1$, where the eligible individuals in state $s_1$ receive the treatment at an infinitesimal amount of time after $t_0$\footnote{Note that we do not limit our setting to panel data.
In particular, the individuals for which the outcome is measured may differ across time points.}.
We denote by $Y^{D=d_0}(t)$ and $Y^{D=d_1}(t)$ the potential outcome variables associated with the outcome at time $t$ in the absence, and in the presence of treatment, respectively. 
To improve readability, we will often use the short-hands $\Y(t)$ and $Y^1(t)$ for $Y^{D=d_0}(t)$ and $Y^{D=d_1}(t)$, respectively. 
We denote the observed outcome at time $t$ by $Y(t)$.
Throughout, we make the following standard consistency assumption \citep{rubin1980randomization}.
\begin{assumption}[Consistency]\label{as:cons}
    At each time $t$, the realized outcome $Y(t)$ is determined as
    \[Y(t)=\sum_d\ind{D=d}\cdot Y^{D=d}(t),\]
    where $\ind{\cdot}$ denotes the indicator function.
\end{assumption}

The estimand of interest is the effect of treatment on the treated, i.e., the group corresponding to $S=s_1,D=d_1$.
As such, we target learning the probability measure over the counterfactual outcome $\Y$ within this subgroup:
\[F(\Y\mid S=s_1,\:D=d_1),\]
where $F$ denotes the cumulative density function.
When clear from context, we use the shorthand $F_{\Y\mid s_1,d_1}$ instead.
Throughout, we assume that random variables are defined over a compact domain, and that densities are absolutely continuous with respect to the Lebesgue measure.

\section{One-dimensional Estimator}\label{sec:one}

\import{./figures/}{figure1}

We commence our analysis by considering cases where the outcome of interest, $Y$, is one-dimensional, i.e., a scalar.

In contrast to \citet{athey2006identification}, we posit that the outcome of an individual can be determined by a combination of the latent variable $U$ (with a common support across groups), and the group to which the individual belongs.
This adjustment relaxes the model assumption of the CiC framework, as outlined below.

\begin{customass}{1}[Model assumption]\label{as:model2}
    At each state $S=s$ and treatment group $D=d$,
    the potential outcomes $\Y(t)$ can be determined through a production function $h_{s,d}(\cdot\:;t)$ of a latent variable $U$, which models the individual characteristics:
    \begin{equation*}
        \forall t\in\{0,1\}:\quad \Y(t) = h_{s,d}(U; t).
    \end{equation*}
\end{customass}

    A key distinction between our setup and the classic CiC framework lies in relaxing \ref{as:model} to \ref{as:model2}, allowing the production functions $h_{s,d}(\cdot)$ to be \emph{group-specific}.
    In applications involving targeted treatment assignments, \ref{as:model2} emerges as a more sensible assumption.
Assumptions \ref{as:monotone} and \ref{as:invariance} are adapted analogously:

\begin{customass}{2}[Strict monotonicity]\label{as:monotone2}
    The production functions $h_{s,d}(\,\cdot\,; t)$ are strictly increasing in $U$ for every $t\in\{t_0,t_1\}$, and every $s,d$.
\end{customass}
\begin{customass}{3}[Time invariance]\label{as:invariance2}
    Within every subgroup, the distribution of the latent variable $U$ does not change over time.
    That is, $\forall u, \forall s,d$,
    \[F_{U\vert S=s, D=d, T=t_1}(u) = F_{U\vert S=s, D=d, T=t_0}(u).\]
\end{customass}
Additionally, we prefer to formulate the overlap assumption in relation to the potential outcomes rather than the latent variable $U$. 
This preference arises from the broader accessibility and interpretability of the support of the outcome, as opposed to that of latent characteristics.
\begin{customass}{4}[Outcome support overlap]\label{as:sup2}
    The potential outcomes $\Y(t_0)$ and $\Y(t_1)$ are defined over domains $\mathbb{Y}_0$ and $\mathbb{Y}_1$, respectively, which are common across subgroups $S\in\{s_0,s_1\}, D\in\{d_0,d_1\}$.
    Moreover, $\mathbb{Y}_0\subseteq\mathbb{Y}_1$.
\end{customass}

Assumptions \ref{as:model2} through \ref{as:sup2} establish the existence of four distinct monotone maps that push forward the density of $\Y(t_0)$ to that of $\Y(t_1)$ in each subgroup.
In particular, let $T_{s,d}$ denote the monotone map that pushes forward the density of $\Y(t_0)$ in the group corresponding to $S=s, D=d$ to the density of $\Y(t_1)$ in the same group (see Figure \ref{fig:four}.)
Specifically, $T_{s,d}$ can be expressed in terms of the $h(\cdot)$ functions as\footnote{For the purposes of this section, Eq.~\eqref{eq:td3} can be considered as the definition of maps $T_{s,d}$.}
\begin{equation}\label{eq:td3}
    T_{s,d}(y) = h_{s,d}\big( h_{s,d}^{-1}(y; t_0) ; t_1\big),
\end{equation}
or equivalently, in terms of the cumulative density functions,
\begin{equation}\label{eq:td2}
    T_{s,d} = F^{-1}_{\Y(t_1)\cond s,d}\circ F_{\Y(t_0)\cond s,d}.
\end{equation}
It is noteworthy that under \ref{as:cons}, $F_{\Y(t)\mid s,t}=F_{Y(t)\mid s,t}$ for every subgroup except $\{s_1,d_1\}$.
Consider the map
$T^*_s=T_{s,d_1}\circ T^{-1}_{s,d_0}$.
$T^*_s$ can be interpreted as the non-linear drift between the maps from $\Y(t_0)$ to $\Y(t_1)$ among control and treatment groups in any state $s$ (see Figure \ref{fig:four}.)
Identification in CiC is achieved by assuming that there is no drift, i.e., $T^*_s$ is the identity map (see Eq.~\ref{eq:tid}).
We shall proceed by relaxing this assumption as follows.
\begin{assumption}[State-independent drifts]\label{as:stindep}
The drift between the mappings of potential outcomes $\Y(t_0)$ to $\Y(t_1)$ in the control and treatment groups is independent of the state.
Formally,
\[T^*\coloneqq T_{s_0,d_1}\circ T^{-1}_{s_0,d_0} \equiv T_{s_1,d_1}\circ T^{-1}_{s_1,d_0}.\]
\end{assumption}
Specifically, rather than assuming \emph{no drift}, we have relaxed the assumption to \emph{equal drift} across the two states.
\ref{as:stindep} can also be expressed in terms of the production functions, albeit at the cost of interpretability:
 \begin{multline}\label{eq:stindep}
            h_{s_0,d_1}\Big\{
            h_{s_0,d_1}^{-1}\big\{
            h_{s_0,d_0}\big(
            h_{s_0,d_0}^{-1}(\cdot;t_1)
            ;t_0\big)
            ;t_0\big\}
            ;t_1\Big\}
            =\\
            h_{s_1,d_1}\Big\{
            h_{s_1,d_1}^{-1}\big\{
            h_{s_1,d_0}\big(
            h_{s_1,d_0}^{-1}(\cdot;t_1)
            ;t_0\big)
            ;t_0\big\}
            ;t_1\Big\}.
    \end{multline}
We are now ready to state our identification result.
The complete set of proofs for our results can be found in Appendix \ref{apx:proofs}.

\begin{restatable}{theorem}{thmid}\label{thm:id}
    Under assumptions \ref{as:model2} - \ref{as:sup2} and \ref{as:cons} - \ref{as:stindep}, the cumulative density function of the missing counterfactual $\Y(t_1)$ in the group $S=s_1,D=d_1$ is identified as:
    \begin{equation}\label{eq:thm1}\begin{split}
        F_{\Y(t_1)}&_{\mid s_1,d_1}(y)
        =\\&
        F_{Y(t_0)\mid s_1,d_1}
        \circ
        T^{-1}_{s_0,d_1}
        \circ
        T_{s_0,d_0}
        \circ
        T^{-1}_{s_1,d_0}
        (
        y)
        ,
    \end{split}
    \end{equation}
where $T_{s,d}$ is given by Eq.~\eqref{eq:td2}.
\end{restatable}
\begin{remark}
    To avoid unnecessarily heavy notation, we did not discuss the observed covariates.
    However, an identical analysis can be done after adjusting for the observed covariates, $X$.
    In particular, production functions may depend on the observed covariates, as long as their monotonicity in $U$ is maintained for every $x$ in the domain of $X$. 
    \ref{as:invariance2}, \ref{as:sup2} and \ref{as:stindep} need to be valid conditioned on $X$ in this scenario, and Eq.~\eqref{eq:thm1} must hold when every term is conditioned on $X$.
\end{remark}
\begin{remark}
    We articulated our identifiability assumptions in accordance with the original work of \citet{athey2006identification}.
    An alternative way of presenting the assumptions would be to do it akin to the \emph{quantile-quantile equi-confounding bias} assumption proposed by \citet{ghassami2022combining}.
    In our context, this would translate to directly assuming the existence of monotone maps $T_{s,d}$ based on Eq.~\eqref{eq:td2} instead of drawing conclusions from \ref{as:model2}-\ref{as:sup2} to establish it.
\end{remark}
\begin{remark}
    We formulated production functions to model the potential outcomes under no treatment ($\Y$), whereas we left the other potential outcome, $Y^1$, unrestricted.
    Due to symmetry, one could model $Y^1$ using production functions and leave $\Y$ unrestricted.
    This scenario might arise for instance, in an study where 3 out of 4 cohorts receive treatment.
    One should exercise greater caution in such cases however, since in practice, the treatment may have effects that significantly alter the composition of the population under study.
    Under these circumstances, assuming monotone production functions for $Y^1$ could be a more drastic assumption.
\end{remark}
\subsection{Relaxing monotonicity}
As mentioned earlier, \ref{as:monotone2} (and its counterpart, \ref{as:monotone} in \citet{athey2006identification}) is an untestable assumption, and might be drastic to impose in certain applications.
The monotonicity of functions $h_{s,d}(\cdot;t)$ in $U$ has two implications:
(i) these functions are bijective, establishing a well-defined inverse for them;
(ii) the property that $\mathbbm{P}(h_{s,d}(U;t)\leq y)=\mathbbm{P}(U\leq h^{-1}_{s,d}(y;t))$, which is repeatedly utilized in proving the point identification result in Theorem \ref{thm:id} (see Appendix \ref{apx:proofs}).
While the bijectivity of  $h_{s,d}(\cdot;t)$ appears to be essential for our framework to work, we can relax (ii).
~Let us first rephrase the monotonicity assumption in the equivalent form:
\begin{equation}\label{eq:rephmono}
    \langle u_0- u_1, h_{s,d}(u_0;t)- h_{s,d}(u_1;t)\rangle\geq0,\quad \forall u_0,u_1.
\end{equation}
This assumption can be relaxed as follows.
\begin{assumption}[$\epsilon$-monotonicity]\label{as:asm}
    Functions $h^{-1}_{s,d}(\cdot;t)$ are well-defined. Additionally, for any $u$ in the support of $U$, 
    \begin{equation}\label{eq:asmonotone}\p\big(\langle U- u, h_{s,d}(U;t)- h_{s,d}(u;t)\rangle<0\mid s,d\big)\leq\frac{\epsilon}{2}.\end{equation}
\end{assumption}
For example, the monthly income of an individual in terms of her age after adjusting for the other covariates can be a $\epsilon$-monotone function.
In general, monthly income increases due to promotions and inflation.
However, temporary unemployment and retirement can affect this trend.
See Figure \ref{fig:eps} for a visualization.

\begin{figure}
    \centering
    \includegraphics[width=0.6\textwidth]{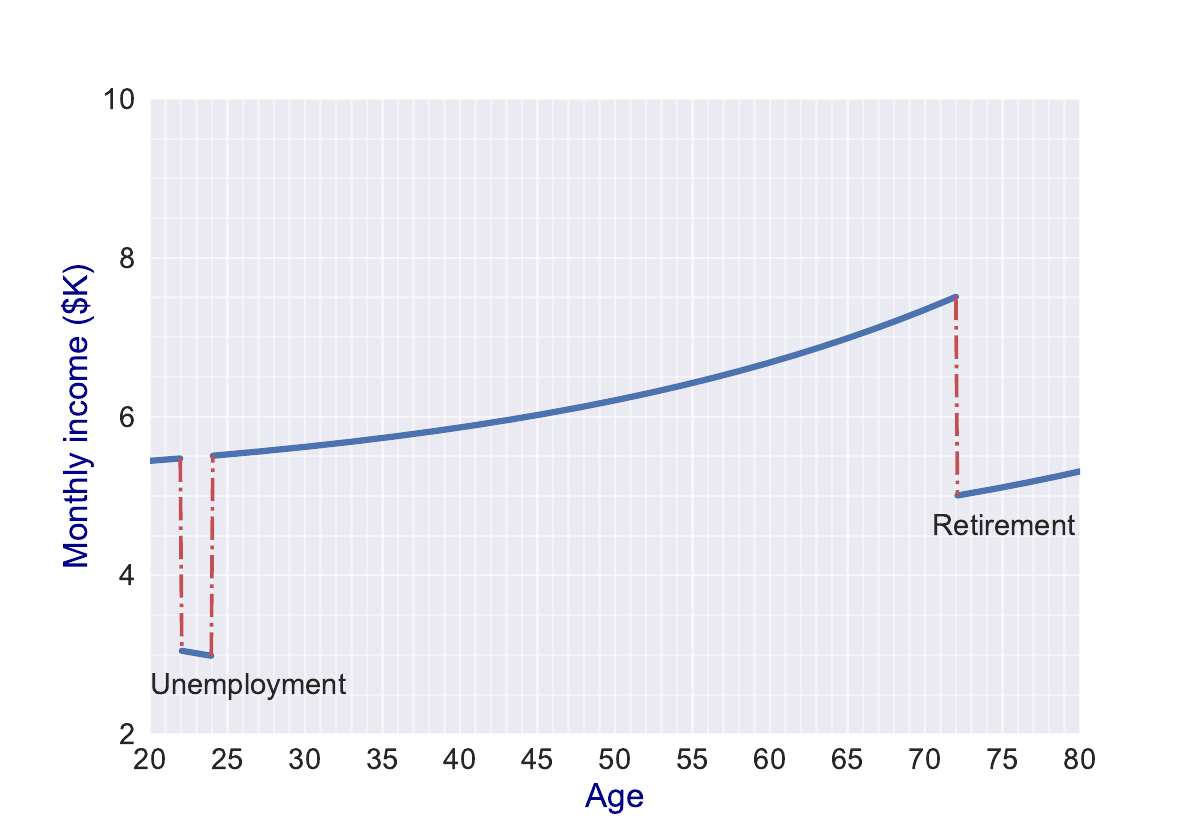}
    \caption{Gross monthly income versus age after adjusting for other covariates. }
    \label{fig:eps}
\end{figure}

We expect that under this relaxation, point identification cannot be achieved.
However, the following partial identification result holds.
\begin{restatable}{proposition}{prppartial}\label{prp:partial}
Under assumptions \ref{as:model2}, \ref{as:invariance2}, \ref{as:sup2} and \ref{as:cons} - \ref{as:asm}, for any $y\in\mathbb{Y}_1$,
\begin{equation*}
\begin{split}
        F_{Y(t_0)\mid s_1,d_1}&\circ
        \mathrm{\underline{\Phi}}^{-1}_{s_0,d_1}
        \circ
        \mathrm{\underline{\Phi}}_{s_0,d_0}
        \circ
        \mathrm{\underline{\Phi}}^{-1}_{s_1,d_0}
        (
        y)
        -\epsilon
        \\&\hspace{3cm}\leq
        F_{\Y(t_1)\mid s_1,d_1}(
        y)
        \leq\\
        &\hspace{4.5cm}F_{Y(t_0)\mid s_1,d_1}\circ
        \mathrm{\overline{\Phi}}^{-1}_{s_0,d_1}
        \circ
        \mathrm{\overline{\Phi}}_{s_0,d_0}
        \circ
        \mathrm{\overline{\Phi}}^{-1}_{s_1,d_0}
        (
        y)
        +\epsilon,
\end{split}
\end{equation*}
where 
\[
\begin{split}
\mathrm{\underline{\Phi}}_{s,d}(y)=F^{-1}_{Y(t_1)\vert s,d}\big(F_{Y(t_0)\vert s,d}(y)-\epsilon\big),\\
\mathrm{\underline{\Phi}}^{-1}_{s,d}(y)=F^{-1}_{Y(t_0)\vert s,d}\big(F_{Y(t_1)\vert s,d}(y)-\epsilon\big),\\
\mathrm{\overline{\Phi}}_{s,d}(y)=F^{-1}_{Y(t_1)\vert s,d}\big(F_{Y(t_0)\vert s,d}(y)+\epsilon\big),\\
\mathrm{\overline{\Phi}}^{-1}_{s,d}(y)=F^{-1}_{Y(t_0)\vert s,d}\big(F_{Y(t_1)\vert s,d}(y)+\epsilon\big).
\end{split}
\]
\end{restatable}
\begin{remark}
    Note that strict monotonicity is a special case of \ref{as:asm}, which corresponds to $\epsilon=0$.
    Accordingly, Proposition \ref{prp:partial} reduces to Theorem \ref{thm:id} with the choice of $\epsilon=0$.
\end{remark}
\begin{remark}
    An analogous result applies to the original framework of CiC when \ref{as:monotone} is relaxed to $\epsilon$-monotonicity.
    See Appendix \ref{apx:CiC} for details.
\end{remark}
\subsection{Relaxing time invariance}
In certain applications, \ref{as:invariance2} may also be violated.
In particular, in repeated cross-sections, it might be challenging to maintain the same distribution among the cases under study.
Even if possible, this may both reduce the number of available samples and result in selection bias, making inference more complicated.
As such, we consider a relaxation of \ref{as:invariance2} and derive a partial identification result in this setting.
We assume that the distribution of the latent variable may change over time, but this change is bounded in Kolmogorov (aka KS) distance \cite{kolmogorov1933sulla}.
\begin{restatable}[$\delta$-invariance]{assumption}{asdelta}\label{as:delta}
    Within every subgroup, the Kolmogorov distance of the distribution of the latent variable across $T=t_0$ and $T=t_1$ is bounded by $\delta$.
     That is,
    \[\sup_u\big\vert F_{U\vert S=s,D=d,T=t_1}(u)-F_{U\vert S=s,D=d,T=t_0}(u)\big\vert\leq\delta.\]
\end{restatable}
Note again that when $\delta=0$, \ref{as:delta} reduces to \ref{as:invariance2}.
\begin{restatable}{proposition}{prpdelta}\label{prp:delta}
    Under assumptions \ref{as:model2}, \ref{as:monotone2}, \ref{as:delta}, \ref{as:sup2}, and \ref{as:cons}-\ref{as:stindep}, for any $y\in\mathbb{Y}_1$,
\begin{equation*}\small
\begin{split}
        F_{Y(t_0)\mid s_1,d_1}&\circ
        \mathrm{\underline{\Psi}}^{-1}_{s_0,d_1}
        \circ
        \mathrm{\underline{\Psi}}_{s_0,d_0}
        \circ
        \mathrm{\underline{\Psi}}^{-1}_{s_1,d_0}
        (
        y)
        -\delta
        \\&\leq
        F_{\Y(t_1)\mid s_1,d_1}(
        y)
        \leq\\
        &F_{Y(t_0)\mid s_1,d_1}\circ
        \mathrm{\overline{\Psi}}^{-1}_{s_0,d_1}
        \circ
        \mathrm{\overline{\Psi}}_{s_0,d_0}
        \circ
        \mathrm{\overline{\Psi}}^{-1}_{s_1,d_0}
        (
        y)
        +\delta,
\end{split}
\end{equation*}
where
\[
\begin{split}
\mathrm{\underline{\Psi}}_{s,d}(y)=F^{-1}_{Y(t_1)\vert s,d}\big(F_{Y(t_0)\vert s,d}(y)-\delta\big),\\
\mathrm{\underline{\Psi}}^{-1}_{s,d}(y)=F^{-1}_{Y(t_0)\vert s,d}\big(F_{Y(t_1)\vert s,d}(y)-\delta\big),\\
\mathrm{\overline{\Psi}}_{s,d}(y)=F^{-1}_{Y(t_1)\vert s,d}\big(F_{Y(t_0)\vert s,d}(y)+\delta\big),\\
\mathrm{\overline{\Psi}}^{-1}_{s,d}(y)=F^{-1}_{Y(t_0)\vert s,d}\big(F_{Y(t_1)\vert s,d}(y)+\delta\big).
\end{split}
\]
\end{restatable}
\begin{remark}
Even more generally, we can simultaneously relax \ref{as:monotone2} to \ref{as:asm} and \ref{as:invariance2} to \ref{as:delta}. 
See Proposition \ref{prp:partialgen} in Appendix \ref{apx:partial} for the partial identification result for this case.
\end{remark}
\subsection{Marginal contrasts and  joint counterfactuals}\label{sec:func}
Theorem \ref{thm:id} guarantees the identification of the probability density of the missing counterfactual, $\Y(t_1)$ in the group corresponding to $S=s_1,D=d_1$.
Having access to this density, we can compute any \emph{marginal contrast estimand} \cite{franks2019flexible}.
A marginal contrast estimand is an estimand that can be expressed as a functional of the marginal distribution of the counterfactual outcomes.
This includes a vast majority of commonly used estimands, such as average treatment effects, conditional average treatment effects, quantile treatment effects, risk ratios, etc.
However, in certain applications, more information is desired.
Examples of estimands that are not marginal contrasts include the \emph{distribution} of the treatment effect, individual-level treatment effects, or the quantiles of the treatment effect.
As a concrete example, consider $Z \coloneqq Y^1(t_1) - \Y(t_1)$.
The density of $Z$ is not identifiable from merely the marginal densities of $Y^1(t_1)$ and $\Y(t_0)$.
In particular, the treatment may have non-zero effects on a fraction of the population even if the two potential outcomes have identical distributions.
Estimands that are not marginal contrasts require stronger identifiability assumptions in general.
In this section, we discuss a stronger version of \ref{as:invariance2} that can lead to the identification of such estimands.

\begin{customassu}{3}[Strong time invariance]\label{as:invariance3}
The latent variable $U$ does not change over time, and there is no loss to follow-up.
\end{customassu}

\ref{as:invariance3} is stronger than \ref{as:invariance2}, in the sense that it completely rules out the possibility of any changes in the latent variable itself, or the cohort of the study.
In contrast, \ref{as:invariance2} would allow for loss to follow-up, as long as similar individuals were recruited for the study.
\ref{as:invariance2} even accommodates a repeated cross-sections study.
\ref{as:invariance2} also allows for the evolution of the latent variable, as long as its distribution among the study population remains the same.
On the other hand, the stronger assumption \ref{as:invariance3} allows for the identification of a wider range of causal estimands.
In particular, under \ref{as:invariance3}, the joint density of the counterfactuals $\big(\Y(t_1), Y^{1}(t_1)\big)$ is identified.

\begin{restatable}{proposition}{prpjoint}\label{prp:joint}
    Under assumptions \ref{as:model2}, \ref{as:monotone2}, \ref{as:invariance3}, \ref{as:sup2}, and \ref{as:cons} - \ref{as:stindep}, the joint density of $\Y(t_1)$ and $Y^1(t_1)$ in the group corresponding to $S=s_1, D=d_1$ is identified as
    \begin{equation*}
    \begin{split}
        &F_{\Y(t_1), Y^1(t_1)\mid s_1,d_1}(y^0,y^1) 
        =\\&
        F_{Y(t_0), Y(t_1)\mid s_1,d_1}\big(F^{-1}_{Y(t_0)\mid s_1,d_1} \circ F_{\Y(t_1)\mid s_1,d_1}(y^0),y^1\big),
    \end{split}\end{equation*}
    where $F_{\Y(t_1)\mid s_1,d_1}(\cdot)$ is given by Equation \eqref{eq:thm1}.
\end{restatable}

\section{Inference}\label{sec:inf}
In this section, we discuss the estimation aspect of our framework with a focus on the average effect of treatment on the treated.
More formally, we consider the estimation of 
\[\tau\coloneqq\E{Y^1(t_1)-\Y(t_1)\vert S=s_1,D=d_1}.\]
Note that under the identifiability assumptions of Section \ref{sec:one}, $\tau$ is identified as
\begin{multline}\label{eq:tau}
\tau = \E{Y(t_1)\vert s_1,d_1} - 
\E{F^{-1}_{Y(t_1)\mid s_0,d_1}
        \circ\\
        F_{Y(t_0)\mid s_0,d_1}
        \circ
        F^{-1}_{Y(t_0)\mid s_0,d_0}
        \circ
        F_{Y(t_1)\mid s_0,d_0}
        \circ\\
        F^{-1}_{Y(t_1)\mid s_1,d_0}
        \circ
        F_{Y(t_0)\mid s_1,d_0}\big(
        Y(t_0)
        \big)\vert s_1,d_1}.
\end{multline}

We make the following assumption on the data generating mechanism to render estimation feasible.
\begin{assumption}\label{as:estimation}
    Conditioned on $S=s,D=d,T=t$, variables $Y(t)$ are continuous random variables defined on a shared bounded domain $[\underline{y},\overline{y}]$, with continuously differentiable density functions $f_{sdt}$, where $f_{sdt}$ is bounded from above and away from $0$, and $\partial f_{sdt}/\partial y$ is bounded.
    For all $s,d,t$, $p_{sdt}=\mathbb{P}(S\!=s,D\!=d,T\!=t)>0$, and given $s,d,t$, the samples $Y_i(t)$ are independent draws from $f_{sdt}$.
\end{assumption}
Let $N$ be the total number of observed outcome samples.
Akin to \citet{athey2006identification}, we build a finite-sample estimator for $\tau$ based on empirical estimators of cumulative density functions.
In particular, let $\{Y_{sd,i}(t)\}_{i=1}^{N_{sdt}}$ denote the independent samples of the outcome at time $t$ in group $S=s, D=d$.
We define
\begin{equation}\label{eq:hatf}
    \hat{F}_{Y(t)\vert s,d}(y) \coloneqq N_{sdt}^{-1}\sum_{i=1}^{N_{sdt}}\ind{Y_{sd,i}(t)\leq y},\quad\text{and,}
\end{equation}
\begin{equation}\label{eq:hatfinv}
    \hat{F}^{-1}_{Y(t)\vert s,d}(u) \coloneqq \inf\{y\in\mathbb{Y}_t, \hat{F}_{Y(t)\vert s,d}(y)\geq u\}.
\end{equation}
Finally, the estimator for $\tau$ is built as:
\begin{multline}\label{eq:hattau}
    \hat{\tau}\coloneqq
    N^{-1}_{s_1d_1t_1}\sum_{i=1}^{N_{s_1d_1t_1}}Y_{s_1d_1,i}(t_1) - \\N^{-1}_{s_1d_1t_0}\sum_{i=1}^{N_{s_1d_1t_0}}
    \hat{F}^{-1}_{Y(t_1)\mid s_0,d_1}
        \circ
        \hat{F}_{Y(t_0)\mid s_0,d_1}
        \circ\\
        \hat{F}^{-1}_{Y(t_0)\mid s_0,d_0}
        \circ
        \hat{F}_{Y(t_1)\mid s_0,d_0}
        \circ
        \hat{F}^{-1}_{Y(t_1)\mid s_1,d_0}
        \circ\\
        \hat{F}_{Y(t_0)\mid s_1,d_0}\big(
        Y_{s_1d_1,i}(t_0)
        \big).
\end{multline}
The following theorem establishes the consistency and asymptotic normality of $\hat{\tau}$.
\begin{restatable}{theorem}{thmconsistency}\label{thm:consistency}
    Under \ref{as:estimation}, $\hat{\tau}-\tau = \mathcal{O}_p(N^{-\frac{1}{2}})$, and \begin{multline}
        \sqrt{N}(\hat{\tau}-\tau)\overset{D}{\rightarrow}\mathcal{N}(0,
    \dfrac{V_0}{p_{s_1d_1t_1}}+
    \dfrac{V_1}{p_{s_0d_1t_1}}+
    \dfrac{V_2}{p_{s_0d_1t_0}}+\\
    \dfrac{V_3}{p_{s_0d_0t_0}}+
    \dfrac{V_4}{p_{s_0d_0t_1}}+
    \dfrac{V_5}{p_{s_1d_0t_1}}+
    \dfrac{V_6}{p_{s_1d_0t_0}}+
    \dfrac{V_7}{p_{s_1d_1t_0}}
    ),
    \end{multline} where $\{V_i\}_{i=0}^7$ are given by Eq.~\eqref{eq:varianceterms}.
\end{restatable}
The expression for the variance and the discussion on its estimation are postponed to Appendix \ref{apx:asymptotic}.
\section{Optimal Transport Representation and High-dimensional Extension} \label{sec:ot}
So far, we focused on the special case of scalar outcome for ease of presentation.
In this section, we discuss the generalization of our results to cover multi-dimensional outcomes.
To this end, we first review the one-dimensional case from an optimal transport point of view.
Let $\eta_{s,d}$ and $\mu_{s,d}$ denote the probability measures over $\Y(t_0)$ and $\Y(t_1)$ in the group corresponding to $S=s,D=d$, respectively.
Brenier's theorem \cite{brenier1991polar} implies that there exists a unique mapping $T$ such that $T_\#\eta_{s,d}=\mu_{s,d}$, i.e., $T$ pushes forward $\eta_{s,d}$ to $\mu_{s,d}$, and $T$ is the gradient of a convex function.
Moreover, $T$ is the optimal transport map with quadratic cost.
More precisely, let $\Gamma$ denote the space of joint distributions over $\big(\Y(t_0),\Y(t_1)\big)$ conditioned on $S=s, D=d$, that agree with the marginal densities $\eta_{s,d}$ and $\mu_{s,d}$.
The optimization problem
\begin{equation}
    \inf_{\gamma\in\Gamma}\int\vert\vert y_0-y_1\vert\vert^2d\gamma(y_0,y_1)
\end{equation}
has a unique solution $\gamma^*$, where
    $\big(\Y(t_0),\Y(t_1)\big)\sim \gamma^*$ if and only if $\Y(t_0)\sim \eta_{s,d}$ and $\Y(t_1) = T\big(\Y(t_0)\big)$, $\eta_{s,d}-a.s$.\footnote{Note that we have omitted the implicit conditioning on $S=s,D=d$ in our notation to improve readability.}
It is straightforward to verify that a one-dimensional function is the gradient of a convex function if and only if it is monotone.
Therefore, the monotonicity of $h_{s,d}(\cdot)$ and the monotonicity of $T_{s,d}(\cdot)$ (as a consequence of the latter) imply $T_{s,d}\equiv T$.
In other words, the mapping $T_{s,d}$ is identified as the optimal transport map with quadratic cost that pushes forward $\eta_{s,d}$ to $\mu_{s,d}$.

Indeed, monotonicity can be slightly relaxed.
The identifiability of the map $T_{s,d}$ in one dimension is guaranteed under \emph{co-monotonicity} of $h_{s,d}(\cdot;t_0)$ and $h_{s,d}(\cdot;t_1)$:
\begin{equation}\label{eq:comonotone}
    \langle u_0-u_1, h_{s,d}(u_0;t_0)-h_{s,d}(u_1;t_1)\rangle\geq 0, \:\: \forall u_0,u_1.
\end{equation}
In order to achieve identifiability results in higher dimensions based on Brenier's theorem, we need to make sure that $T_{s,d}$ is the gradient of a convex function.
It is known that a function is the gradient of a convex function if and only if it is \emph{cyclically monotone} \cite{rockafellar1970convex}.
A function $h$ is said to be cyclically monotone if for any sequence $x_0,\dots,x_n$ in its domain,
\begin{equation}\label{eq:cyclic}
\sum_{i=0}^n\langle h(x_i), x_i-x_{i+1}\rangle\geq0,\end{equation}
where $x_{n+1}=x_0$.
For $n=1$, Eq.~\eqref{eq:cyclic} reduces to Eq.~\eqref{eq:rephmono}.
With this preliminary discussion in place, the identification result in higher dimensions can be stated as follows.
\begin{customassu}{2}[Co-cyclic monotonocity]\label{as:cocyclic}
    Functions $h_{s,d}(\cdot;t_0)$ and $h_{s,d}(\cdot;t_1)$ are co-cyclically monotone.
    That is, for any sequence $u_0,\dots, u_n$ in their common domain,
    \[\sum_{i=0}^n\langle h_{s,d}(u_i;t_0), h_{s,d}(u_i;t_1)-h_{s,d}(u_{i+1};t_1)\rangle\geq0,\]
    where $u_{n+1}=u_0$.
\end{customassu}
\begin{restatable}{theorem}{thmhighd}\label{thm:highd}
    Under assumptions \ref{as:model2}, \ref{as:cocyclic}, \ref{as:invariance2}, \ref{as:sup2}, and \ref{as:cons} - \ref{as:stindep}, the probability measure over the missing counterfactual, i.e., $\mu_{s_1,d_1}$, is identified as
    \begin{equation}
        \mu_{s_1,d_1} = (T^*\circ T_{s_1,d_0})_\#\eta_{s_1,d_1},
    \end{equation}
    where $T_{s,d}$ is the Brenier map that pushes forward $\eta_{s,d}$ to $\mu_{s, d}$ for every $s\in\{s_0,s_1\}, d\in\{d_0,d_1\}$, and $T^*$ is the Brenier map that pushes forward $(T_{s_0,d_0\#}\eta_{s_0,d_1})$ to $\mu_{s_0,d_1}$.
\end{restatable}
Theorem \ref{thm:highd} is in essence a generalization of Theorem \ref{thm:id}.
To see this, note that the Brenier map that pushes forward $\eta_{s,d}$ to $\mu_{s,d}$ in one dimension is precisely $F^{-1}_{\Y(t_1)\vert s,d}\circ F_{\Y(t_0)\vert s,d}$.
However, the co-cyclic monotonicity assumption \ref{as:cocyclic} is not as easy to interpret as \ref{as:monotone2}.
To address this challenge, we follow the proposition of \citet{torous2021optimal} based on the following result.
\begin{proposition}[\citealp{saks2005weak}]\label{prp:saks}
    Let $K$ and $F$ be a convex and a finite subset of $\mathbb{R}^d$ respectively. A function $T:K\to F$ is cyclically monotone if and only if it is monotone.
\end{proposition}
Proposition \ref{prp:saks} implies that in our setting, if the densities $\eta_{s,d}$ are supported on convex sets and $\Y(t_1)$ is finite-valued in every subgroup, then \ref{as:cocyclic} reduces to co-monotonicity \eqref{eq:comonotone}.
In other words, restricting $\eta_{s,d}$ and $\mu_{s,d}$ densities to be defined over convex and finite sets, respectively, Assumption \ref{as:cocyclic} of Theorem \ref{thm:highd} can be replaced by the easily interpretable assumption of Eq.~\eqref{eq:comonotone}.
\begin{figure*}[t]
    \centering
    \begin{subfigure}[b]{0.375\linewidth}
        \includegraphics[width=\textwidth]{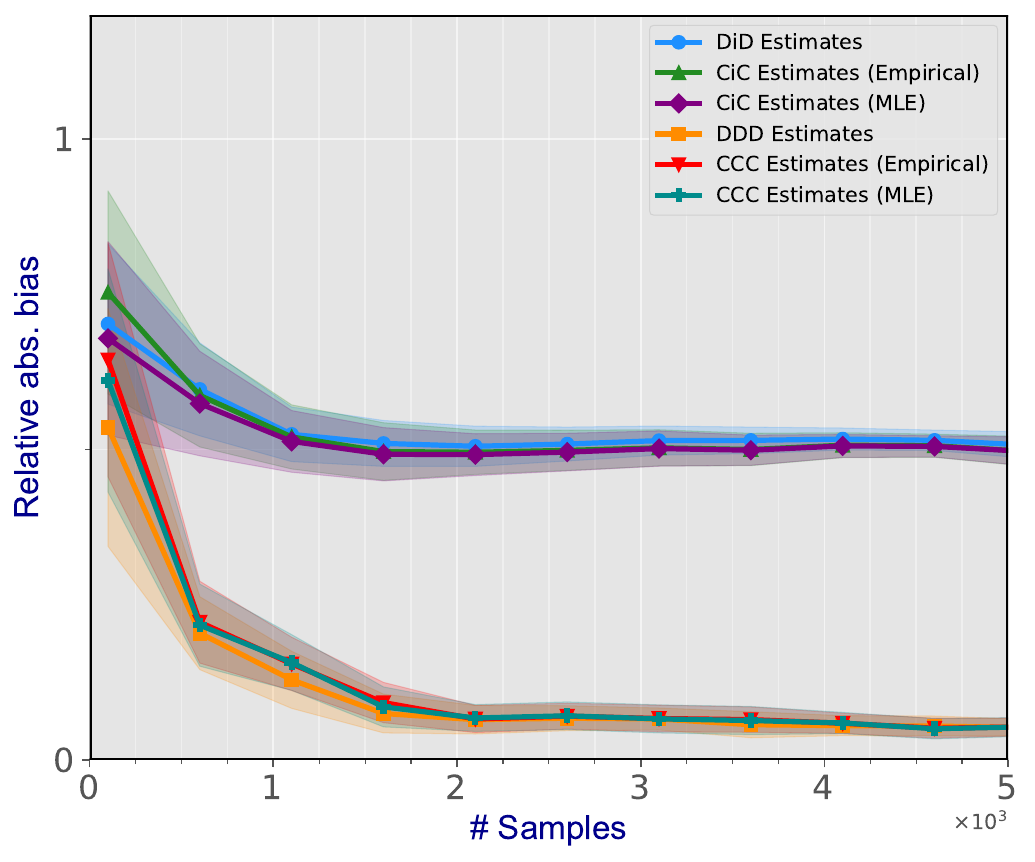}
        \caption{Linear production functions.}
        \label{fig:syn-linear}
    \end{subfigure}\hspace{.4cm}
    \begin{subfigure}[b]{0.405\linewidth}
        \includegraphics[width=\textwidth]{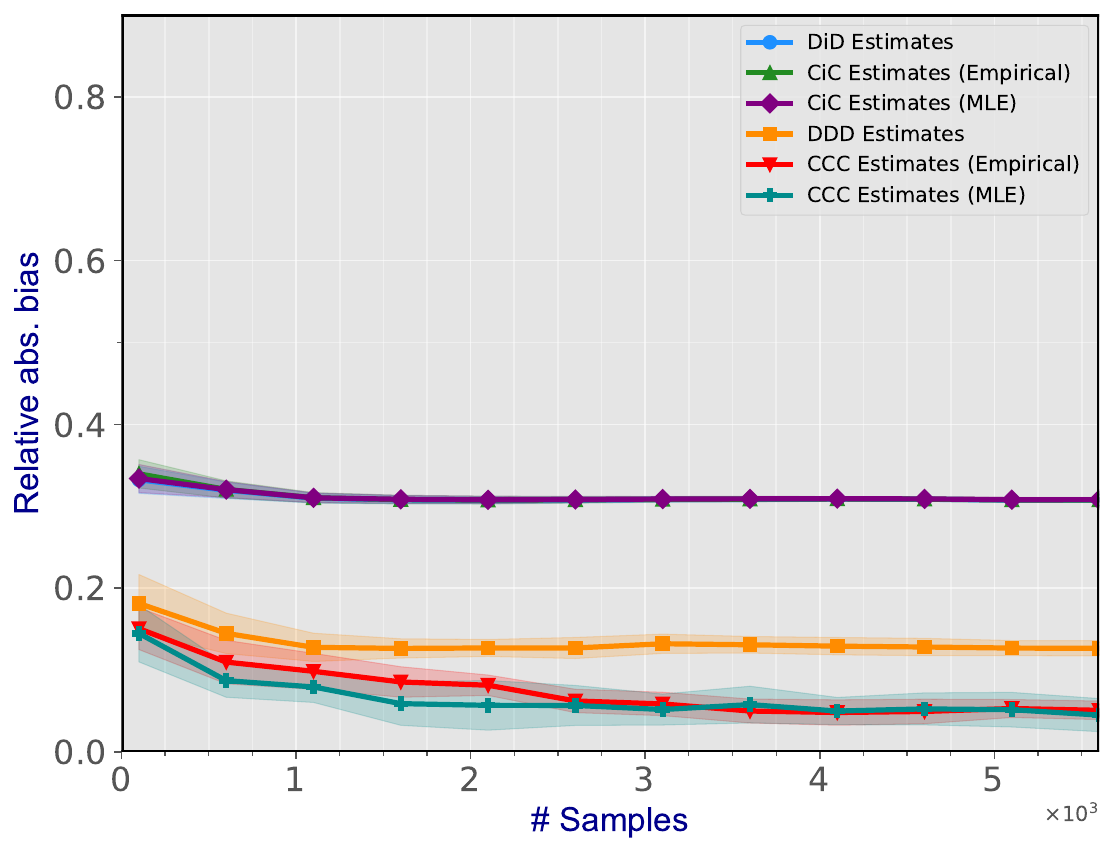}
        \caption{Nonlinear production functions.}
        \label{fig:syn-nonlinear}
    \end{subfigure}
    \caption{Relative bias of the estimators evaluated on (\subref{fig:syn-linear}) a linear, and (\subref{fig:syn-nonlinear}) a non-linear model.}
    \label{fig:syn}
\end{figure*}
\section{Simulation Studies}\label{sec:exp}
Our empirical analysis is structured into two main sections\footnote{
The code to reproduce the results of this paper are accessible at https://github.com/SinaAkbarii/Triple-Changes.}. 
In the first part, we compare the estimation error of the triple changes estimator and DiD, triple difference, and CiC estimators using synthetically generated datasets, where we know the ground truth: i.e., the treatment effect on the treated.
To evaluate the performance of each estimator, we focus on the relative bias metric, defined as
$\varepsilon=\big\vert 1-\frac{\widehat{\tau}}{\tau}\big\vert,$
where 
$\tau$ is the average treatment effect on the treated, given by \eqref{eq:tau}, and $\widehat{\tau}$ represents the estimate provided by the respective estimator.
In the second part, we apply our estimator to data from \emph{National Survey on Children's Health} (NSCH)\footnote{https://mchb.hrsa.gov/national-survey-childrens-health-questionnaires-datasets-supporting-documents} to assess the effect of Medicaid expansion under Affordable Care Act (ACA) in the united states on preventive care for children.

\subsection{Synthetic data}
We begin with a simple linear model where the latent variable $U$ given $s,d$ is sampled from a Gaussian distribution with mean $\nu_{s,d}$ and variance $1$ (see Appendix \ref{apx:exp} for a complete table of parameters $\nu_{s,d}$, as well as further details of our experiment setup,) and the production functions are defined as
\begin{equation}\label{eq:prod}
    h_{s,d}(u;t)\coloneqq 2u+\big(\frac{1+s}{4}+\frac{d-0.5}{2}\big) t.
\end{equation}
The actual outcome at time $t_1$ in the treated group $(s_1,d_1)$ is also sampled from a Gaussian distribution with mean $2.75$ and variance $1$.

CiC and triple changes estimators use estimates of the cumulative density functions and their inverses.
We implemented two versions of each of these estimators, namely a non-parametric estimator using the empirical estimator of density functions (see Eq.~\ref{eq:hatf} and \ref{eq:hatfinv},) and a model-based version using the maximum likelihood estimator given a class of distributions.
The class of distributions in this section are specified correctly, i.e., the true model lies in the class.
See appendix \ref{apx:exp} for a setting where this class is misspecified.

In Figure \ref{fig:syn}, the triple difference and triple changes estimators are denoted by DDD and CCC, respectively.
As depicted by Figure \ref{fig:syn-linear}, DiD and CiC estimators exhibit persistent bias, not converging to zero even with increasing sample size.
In contrast, the biases of triple difference and triple changes estimators approach zero with as the sample size grows.
Furthermore, since the triple difference estimator only relies on empirical averages, it shows slightly lower bias in small sample sizes.

To add non-linearities to the previous model, we modified the data generating mechanism, including two of the production functions.
Specifically, for $(s_0,d_1,t_1)$, and $(s_1,d_1,t_1)$, the production functions were modified to
\[h_{s,d}(u;t) = 0.1\exp\Big(2u+\big(\frac{1+s}{4}+\frac{d-0.5}{2}\big) t\Big),\]
whereas the other six groups were generated according to Eq.~\eqref{eq:prod}.
The outcomes for this model are illustrated in Figure \ref{fig:syn-nonlinear}. Notably, under the nonlinear model, the triple-difference estimator yields biased estimates, while the triple-changes estimators remain (asymptotically) unbiased. Furthermore, owing to the complexity of the model, empirical estimators of density functions exhibit more bias compared to their MLE-based counterparts, particularly in low-sample regimes.

\subsection{Application to NSCH data}
The Medicaid expansion under Affordable Care Act was designed to extend Medicaid coverage to more low-income citizens in the US.
This expansion raised the income threshold for Medicaid eligibility to 138\% of the federal poverty level (FPL).
While Louisiana ($s_1$) adopted this expansion in 2016, the neighboring states of Mississippi and Texas ($s_0$) are yet to do so.
We utilized publicly available anonymous NSCH data for the years 2016 ($t_0$) and 2017 ($t_1$) to assess the impact of Medicaid expansion on children's access to preventive healthcare, specifically analyzing the change in the frequency of doctor visits. Individuals with FPL $\leq100\%$ were classified as eligible for the expansion ($d_1$), while those with FPL $\geq140\%$ were considered not eligible ($d_0$). Data within the uncertainty margin between these thresholds was excluded.

Applying the triple difference and triple changes estimators on the frequency of doctor visits for preventive purposes with 1000 bootstraps resulted in means of $0.170$ and $0.145$, respectively, with $90\%$ confidence intervals of $[-0.004, 0.344]$ and $[-0.01, 0.331]$, respectively.
Our findings suggest that the adoption of Medicaid expansion in Louisiana increased the likelihood of children undergoing an annual preventive care visit, compared to those in non-expansion states. This aligns with the conclusion of \citet{roy2020impact}.

\section{Concluding Remarks}
We formally presented the triple changes estimator for assessing the treatment effect on the treated in observational studies.
We derived a set of necessary assumptions for point identification and discussed partial identification results under relaxed versions of these assumptions. 
Future research avenues could include exploring partial identification under alternative assumptions, investigating extensions to time series settings, and conducting statistical analyses to enhance the robustness and applicability of the estimator.

\bibliographystyle{plainnat}
\bibliography{biblio}
\clearpage

\appendix
\onecolumn
\begin{center}
    \Large\bfseries Appendices
\end{center}

The appendices are organized as follows.
We begin with providing a general partial identification result that was promised in the main text in Section \ref{apx:partial}.
We review the consistency and asymptotic normality of the estimator $\hat{\tau}$ given by Eq.~\eqref{eq:hattau} in Appendix \ref{apx:asymptotic}.
Appendix \ref{apx:lemmas} provides twelve lemmas that form the basis of the proof of Theorem \ref{thm:consistency}.
Appendix \ref{apx:proofs} includes the omitted proofs of the claims made in the main text.
Some of these claims are corollaries of the general result provided in Appendix \ref{apx:partial}.
Finally, Appendix \ref{apx:exp} provides further experimental results, as well as complementary details of the experimental setup used in the main text.
\section{Supplementary Findings on Partial Identification of the Missing Potential Outcome}\label{apx:partial}
Propositions \ref{prp:partial} and \ref{prp:delta} given in the text are in effect corollaries of the result given below.
We later provide an analogous result for the classic CiC framework.
For completeness, we include the relaxed assumptions here.
\begin{assumption}[$\epsilon$-monotonicity]\label{as:eps}
    Functions $h^{-1}_{s,d}(\cdot;t)$ are well-defined, and for any $u$ in the support of $U$, 
    \begin{equation}\label{eq:aseps}\p\big(\langle U- u, h_{s,d}(U;t)- h_{s,d}(u;t)\rangle<0\mid S=s,D=d,T=t\big)\leq\frac{\epsilon}{2}.\end{equation}
\end{assumption}
\begin{remark}
    \ref{as:eps} is slightly different from \ref{as:asm} in that we have also conditioned on $T=t$ in Eq.~\eqref{eq:aseps}.
    Clearly, under \ref{as:invariance2}, the two are equivalent.
    However, we are relaxing \ref{as:invariance2} to \ref{as:delta} in this section, and it is imperative to highlight this difference.
\end{remark}
\asdelta*

\begin{proposition}\label{prp:partialgen}
    Under assumptions \ref{as:model2}, \ref{as:eps}, \ref{as:delta}, \ref{as:sup2}, and \ref{as:cons}-\ref{as:stindep}, for any $y\in\mathbb{Y}_1$,
    \begin{equation*}
        \begin{split}
            F_{Y(t_0)\mid s_1,d_1}\circ
            \mathrm{\underline{\Psi}}^{-1}_{s_0,d_1}
            \circ
            \mathrm{\underline{\Psi}}_{s_0,d_0}
            \circ
            \mathrm{\underline{\Psi}}^{-1}_{s_1,d_0}&
            (
            y)
            -\epsilon-\delta
            \\&\leq
            F_{\Y(t_1)\mid s_1,d_1}(
            y)
            \leq\\
            &\hspace{2cm}F_{Y(t_0)\mid s_1,d_1}\circ
            \mathrm{\overline{\Psi}}^{-1}_{s_0,d_1}
            \circ
            \mathrm{\overline{\Psi}}_{s_0,d_0}
            \circ
            \mathrm{\overline{\Psi}}^{-1}_{s_1,d_0}
            (
            y)
            +\epsilon+\delta,
        \end{split}
        \end{equation*}
        where
        \[
        \begin{split}
            \mathrm{\underline{\Psi}}_{s,d}(y)=F^{-1}_{Y(t_1)\vert s,d}\big(F_{Y(t_0)\vert s,d}(y)-\epsilon-\delta\big),\\
            \mathrm{\underline{\Psi}}^{-1}_{s,d}(y)=F^{-1}_{Y(t_0)\vert s,d}\big(F_{Y(t_1)\vert s,d}(y)-\epsilon-\delta\big),\\
            \mathrm{\overline{\Psi}}_{s,d}(y)=F^{-1}_{Y(t_1)\vert s,d}\big(F_{Y(t_0)\vert s,d}(y)+\epsilon+\delta\big),\\
            \mathrm{\overline{\Psi}}^{-1}_{s,d}(y)=F^{-1}_{Y(t_0)\vert s,d}\big(F_{Y(t_1)\vert s,d}(y)+\epsilon+\delta\big).
        \end{split}
        \]
\end{proposition}
\begin{proof} Suppose $t\in\{t_0,t_1\}$.
\begin{equation}\label{eq:prpgen1p}
    \begin{split}
        F_{\Y(t)\mid s,d}(y)
        &=\mathbbm{P}(h_{s,d}(U;t)\leq y\mid S=s,D=d, T=t)\\
        &=\mathbbm{P}(h_{s,d}(U;t)\leq y, \:U\leq h^{-1}_{s,d}(y;t)\mid S=s,D=d, T=t)
        \\&+\mathbbm{P}(h_{s,d}(U;t)\leq y, \:U>h^{-1}_{s,d}(y;t)\mid S=s,D=d, T=t),
    \end{split}
\end{equation}
which implies the following inequalities:
\begin{equation}\label{eq:prpgen1p1}
    \begin{split}
        F_{\Y(t)\mid s,d}(y)
        \\&\overset{(a)}{\leq}\mathbbm{P}(U\leq h^{-1}_{s,d}(y;t)\mid S=s,D=d, T=t)
        \\&+\mathbbm{P}(h_{s,d}(U;t)\leq y, \:U>h^{-1}_{s,d}(y;t)\mid S=s,D=d, T=t)
        \\&\overset{(b)}{\leq}\mathbbm{P}(U\leq h^{-1}_{s,d}(y;t)\mid S=s,D=d, T=t)+\frac{\epsilon}{2}
        \\&=F_{U\mid s,d,t}(h_{s,d}^{-1}(y;t))+\frac{\epsilon}{2},
    \end{split}
\end{equation}
where $(a)$ is a basic probability manipulation and $(b)$ follows from \ref{as:asm}, and
\begin{equation}\label{eq:prpgen1p2}
    \begin{split}
        F_{\Y(t)\mid s,d}(y)
        \\&\overset{(c)}{\geq}\mathbbm{P}(h_{s,d}(U;t)\leq y, \:U\leq h^{-1}_{s,d}(y;t)\mid S=s,D=d, T=t)
        \\&=\mathbbm{P}(U\leq h^{-1}_{s,d}(y;t)\mid S=s,D=d, T=t)
        \\&-\mathbbm{P}(h_{s,d}(U;t)> y, \:U\leq h^{-1}_{s,d}(y;t)\mid S=s,D=d, T=t)
        \\&\overset{(d)}{\geq}
        \mathbbm{P}(U\leq h^{-1}_{s,d}(y;t)\mid S=s,D=d, T=t)
        -\frac{\epsilon}{2}
        \\&=F_{U\mid s,d,t}(h_{s,d}^{-1}(y;t))-\frac{\epsilon}{2},
    \end{split}
\end{equation}
where $(c)$ is a consequence of \eqref{eq:prpgen1p}, and $(d)$ follows from \ref{as:asm}.
To summarize, we got
\begin{equation}\label{eq:prpgen1pp}
    F_{U\mid s,d,t}(h_{s,d}^{-1}(y;t))-\frac{\epsilon}{2}\leq
    F_{\Y(t)\mid s,d}(y)
    \leq F_{U\mid s,d,t}(h_{s,d}^{-1}(y;t))+\frac{\epsilon}{2},
\end{equation}
and using \ref{as:delta},
\begin{equation}\label{eq:prpgen1ppnew}
    F_{U\mid s,d,\tilde{t}}(h_{s,d}^{-1}(y;t))-\frac{\epsilon}{2}-\delta
    \leq
    F_{\Y(t)\mid s,d}(y)
    \leq F_{U\mid s,d,\tilde{t}}(h_{s,d}^{-1}(y;t))+\frac{\epsilon}{2}+\delta,
\end{equation}
where $\tilde{t}\in\{t_0,t_1\}$.
Eq.~\eqref{eq:prpgen1ppnew} can also be written as
\begin{equation}\label{eq:prpgen1pp2}
    F_{\Y(t)\mid s,d}(y)-\frac{\epsilon}{2}-\delta
    \leq 
    F_{U\mid s,d,\tilde{t}}(h_{s,d}^{-1}(y;t))
    \leq 
    F_{\Y(t)\mid s,d}(y)+\frac{\epsilon}{2}+\delta,
\end{equation}
and consequently,
\begin{equation}\label{eq:prpgen2p}
    F_{U\mid s,d, \tilde{t}}^{-1}\big(F_{\Y(t)\mid s,d}(y)-\frac{\epsilon}{2}-\delta
    \big)\leq h_{s,d}^{-1}(y;t) 
    \leq 
    F_{U\mid s,d,\tilde{t}}^{-1}\big(F_{\Y(t)\mid s,d}(y)+\frac{\epsilon}{2}+\delta
    \big).
\end{equation}
Also, choosing $y=h_{s,d}(u;t)$ in Eq.~\eqref{eq:prpgen1pp}, and applying $F^{-1}_{\Y(t)\mid s,d}(\cdot)$ to its sides, and renaming $t$ to $\tilde{t}$:
\begin{equation}\label{eq:prpgen3p}
    \begin{split}
        F^{-1}_{\Y(\tilde{t})\mid s,d}\big(F_{U\mid s,d,\tilde{t}}(u)-\frac{\epsilon}{2}\big)
        \leq
        h_{s,d}(u;\tilde{t}) \leq F^{-1}_{\Y(\tilde{t})\mid s,d}\big(F_{U\mid s,d,\tilde{t}}(u)+\frac{\epsilon}{2}\big)
    \end{split}
\end{equation}
Finally, choose $u=h_{s,d}^{-1}(y;t)$ in Eq.~\eqref{eq:prpgen3p}, and utilize Eq.~\eqref{eq:prpgen2p} to get
\begin{equation}\label{eq:prpgenineq}
    F^{-1}_{\Y(\tilde{t})\mid s,d}
    \big(
    F_{\Y(t)\mid s,d}(y)
    -\epsilon-\delta
    \big)
    \leq
    h_{s,d}\big(
    h_{s,d}^{-1}(y;t)
    ;\tilde{t}\big)
    \leq
    F^{-1}_{\Y(\tilde{t})\mid s,d}
    \big(
    F_{\Y(t)\mid s,d}(y)
    +\epsilon+\delta
    \big).
\end{equation}
Note that we utilized the fact that CDFs and their inverses are increasing.

In \ref{as:stindep} (see Eq.~\ref{eq:stindep}), choose $y=h_{s_1,d_0}\big(
            h_{s_1,d_0}^{-1}(y';t_0)
            ;t_1\big)$,
and take inverses
which results in
\begin{equation}
    \begin{split}
        h_{s_1,d_1}\big(
        h_{s_1,d_1}^{-1}(y;t_1)
        ;t_0\big)
        =
        h_{s_0,d_1}\Bigg(
        h^{-1}_{s_0,d_1}\bigg(
        h_{s_0,d_0}\Big\{
        h^{-1}_{s_0,d_0}\big\{
        h_{s_1,d_0}\big(
        h^{-1}_{s_1,d_0}(
        y;t_1)
        ; t_0\big)
        ; t_0\big\}
        ; t_1\Big\}
        ; t_1\bigg)
        ;t_0\Bigg),
    \end{split}
\end{equation}
Using the inequalities of \eqref{eq:prpgenineq},
\begin{equation}
\small
    \begin{split}
        &F^{-1}_{\Y(t_0)\mid s_0,d_1}
        \Bigg(
        F_{\Y(t_1)\mid s_0,d_1}\bigg(
        F^{-1}_{\Y(t_1)\mid s_0,d_0}
        \Big\{
        F_{\Y(t_0)\mid s_0,d_0}\big\{
        F^{-1}_{\Y(t_0)\mid s_1,d_0}
        \big(
        F_{\Y(t_1)\mid s_1,d_0}(
        y)
        -\epsilon-\delta\big)
        \big\}
        -\epsilon-\delta\Big\}
        \bigg)
        -\epsilon-\delta\Bigg)
        \\&\leq
        F^{-1}_{\Y(t_0)\mid s_0,d_1}
        \Bigg(
        F_{\Y(t_1)\mid s_0,d_1}\bigg(
        F^{-1}_{\Y(t_1)\mid s_0,d_0}
        \Big\{
        F_{\Y(t_0)\mid s_0,d_0}\big\{
        h_{s_1,d_0}\big(
        h^{-1}_{s_1,d_0}(
        y;t_1)
        ; t_0\big)
        \big\}
        -\epsilon-\delta\Big\}
        \bigg)
        -\epsilon-\delta\Bigg)
        \\&\leq
        F^{-1}_{\Y(t_0)\mid s_0,d_1}
        \Bigg(
        F_{\Y(t_1)\mid s_0,d_1}\bigg(
        h_{s_0,d_0}\Big\{
        h^{-1}_{s_0,d_0}\big\{
        h_{s_1,d_0}\big(
        h^{-1}_{s_1,d_0}(
        y;t_1)
        ; t_0\big)
        ; t_0\big\}
        ; t_1\Big\}
        \bigg)
        -\epsilon-\delta\Bigg)
        \\&
        \leq 
        h_{s_1,d_1}\big(
        h_{s_1,d_1}^{-1}(y;t_1)
        ;t_0\big)
        \\&
        \leq
        F^{-1}_{\Y(t_0)\mid s_0,d_1}
        \Bigg(
        F_{\Y(t_1)\mid s_0,d_1}\bigg(
        h_{s_0,d_0}\Big\{
        h^{-1}_{s_0,d_0}\big\{
        h_{s_1,d_0}\big(
        h^{-1}_{s_1,d_0}(
        y;t_1)
        ; t_0\big)
        ; t_0\big\}
        ; t_1\Big\}
        \bigg)
        +\epsilon+\delta\Bigg)
        \\&\leq
        F^{-1}_{\Y(t_0)\mid s_0,d_1}
        \Bigg(
        F_{\Y(t_1)\mid s_0,d_1}\bigg(
        F^{-1}_{\Y(t_1)\mid s_0,d_0}
        \Big\{
        F_{\Y(t_0)\mid s_0,d_0}\big\{
        h_{s_1,d_0}\big(
        h^{-1}_{s_1,d_0}(
        y;t_1)
        ; t_0\big)
        \big\}
        +\epsilon+\delta\Big\}
        \bigg)
        +\epsilon+\delta\Bigg)
        \\&\leq
        F^{-1}_{\Y(t_0)\mid s_0,d_1}
        \Bigg(
        F_{\Y(t_1)\mid s_0,d_1}\bigg(
        F^{-1}_{\Y(t_1)\mid s_0,d_0}
        \Big\{
        F_{\Y(t_0)\mid s_0,d_0}\big\{
        F^{-1}_{\Y(t_0)\mid s_1,d_0}
        \big(
        F_{\Y(t_1)\mid s_1,d_0}(
        y)
        +\epsilon+\delta\big)
        \big\}
        +\epsilon+\delta\Big\}
        \bigg)
        +\epsilon+\delta\Bigg),
    \end{split}
\end{equation}
and finally applying \eqref{eq:prpgenineq} to the term $h_{s_1,d_1}\big(
        h_{s_1,d_1}^{-1}(y;t_1)
        ;t_0\big)$
itself, we get the two inequalities
\begin{equation}
    \begin{split}
\small
    &F^{-1}_{\Y(t_0)\mid s_1,d_1}
    \big(
    F_{\Y(t_1)\mid s_1,d_1}(
    y)-\epsilon-\delta
    \big)\leq
    \\
    &F^{-1}_{\Y(t_0)\mid s_0,d_1}
    \Bigg(
    F_{\Y(t_1)\mid s_0,d_1}\bigg(
    F^{-1}_{\Y(t_1)\mid s_0,d_0}
    \Big\{
    F_{\Y(t_0)\mid s_0,d_0}\big\{
    F^{-1}_{\Y(t_0)\mid s_1,d_0}
    \big(
    F_{\Y(t_1)\mid s_1,d_0}(
    y)
    +\epsilon+\delta\big)
    \big\}
    +\epsilon+\delta\Big\}
    \bigg)
    +\epsilon+\delta\Bigg),
    \end{split}
\end{equation}
and
\begin{equation}
    \begin{split}
\small
    &F^{-1}_{\Y(t_0)\mid s_0,d_1}
    \Bigg(
    F_{\Y(t_1)\mid s_0,d_1}\bigg(
    F^{-1}_{\Y(t_1)\mid s_0,d_0}
    \Big\{
    F_{\Y(t_0)\mid s_0,d_0}\big\{
    F^{-1}_{\Y(t_0)\mid s_1,d_0}
    \big(
    F_{\Y(t_1)\mid s_1,d_0}(
    y)
    -\epsilon-\delta\big)
    \big\}
    -\epsilon-\delta\Big\}
    \bigg)
    -\epsilon-\delta\Bigg)
    \\&\leq F^{-1}_{\Y(t_0)\mid s_1,d_1}
    \big(
    F_{\Y(t_1)\mid s_1,d_1}(
    y)+\epsilon+\delta
    \big),
    \end{split}
\end{equation}
which complete the proof by applying $F_{\Y(t_0)\mid s_1,d_1}(\cdot)$ to their sides and applying \ref{as:cons}.

\end{proof}

\subsection{Partial identification for changes-in-changes}\label{apx:CiC}
In this section, we present the partial identification result for the conventional framework of changes-in-changes under relaxed versions of monotonicity (\ref{as:monotone}), and time invariance (\ref{as:invariance}).
Serving as a counterpart to Proposition \ref{prp:partialgen} within our triple changes framework, this result stands on its own merit, given the significance and wide applicability of changes-in-changes.
We begin with stating the relaxed assumption.
\begin{assumption}[$\epsilon$-monotonicity]\label{as:epscic}
    Functions $h^{-1}(\cdot;t)$ are well-defined, and for any $u$ in the support of $U$, and any $d\in\{d_0,d_11\}$,
    \[\p\Big(\langle U- u,h(U;t)- h(u;t)\rangle<0\mid D=d, T=t\Big)\leq\frac{\epsilon}{2}.\]
\end{assumption}
\begin{assumption}[$\delta$-invariance]\label{as:deltacic}
    Within every subgroup, the Kolmogorov distance of the distribution of the latent variable at $T=t_0$ and $T=t_1$ is bounded by $\delta$.
     That is,
    \[\sup_u\big\vert F_{U\vert D=d,T=t_1}(u)-F_{U\vert D=d,T=t_0}(u)\big\vert\leq\delta.\]
\end{assumption}
\begin{proposition}
    Under assumptions \ref{as:model}, \ref{as:deltacic}, \ref{as:epscic}, \ref{as:sup}, and \ref{as:cons}, for any $y$ in the support of $\Y$,
    \begin{multline*}
        F_{Y(t_0)\mid d_1}\Big(F^{-1}_{Y(t_0)\mid d_0}
        \big(
        F_{Y(t_1)\mid d_0}(y)
        -\epsilon-\delta
        \big)\Big) - \epsilon - \delta
        \leq
        F_{\Y(t_1)\mid d_1}(y) 
        \leq
        F_{Y(t_0)\mid d_1}\Big(F^{-1}_{Y(t_0)\mid d_0}
        \big(
        F_{Y(t_1)\mid d_0}(y)
        +\epsilon+\delta
        \big)\Big) + \epsilon+\delta.
    \end{multline*}
\end{proposition}
\begin{proof}Suppose $t\in\{t_0,t_1\}$.
    \begin{equation}\label{eq:p31}
    \begin{split}
        F_{\Y(t)\mid d}(y)
        &=\mathbbm{P}(h(U;t)\leq y\mid D=d,T=t)\\
        &=\mathbbm{P}(h(U;t)\leq y, \:U\leq h^{-1}(y;t)\mid D=d,T=t)
        \\&+\mathbbm{P}(h(U;t)\leq y, \:U>h^{-1}(y;t)\mid D=d,T=t),
    \end{split}
\end{equation}
which implies the following inequalities:
\begin{equation}\label{eq:p32}
    \begin{split}
        F_{\Y(t)\mid d}(y)
        &\overset{(a)}{\leq}\mathbbm{P}(U\leq h^{-1}(y;t)\mid D=d,T=t)
        +\mathbbm{P}(h(U;t)\leq y, \:U>h^{-1}(y;t)\mid D=d,T=t)
        \\&\overset{(b)}{\leq}\mathbbm{P}(U\leq h^{-1}(y;t)\mid D=d,T=t)+\frac{\epsilon}{2}
        \\&=F_{U\mid d,t}(h^{-1}(y;t))+\frac{\epsilon}{2},
    \end{split}
\end{equation}
where $(a)$ is a basic probability manipulation and $(b)$ follows from \ref{as:epscic}, and
\begin{equation}\label{eq:p322}
    \begin{split}
        F_{\Y(t)\mid d}(y)
        &\overset{(c)}{\geq}\mathbbm{P}(h(U;t)\leq y, \:U\leq h^{-1}(y;t)\mid D=d,T=t)
        \\&=\mathbbm{P}(U\leq h^{-1}(y;t)\mid D=d,T=t)
        -\mathbbm{P}(h(U;t)> y, \:U\leq h^{-1}(y;t)\mid D=d,T=t)
        \\&\overset{(d)}{\geq}
        \mathbbm{P}(U\leq h^{-1}(y;t)\mid D=d,T=t)
        -\frac{\epsilon}{2}
        \\&=F_{U\mid d,t}(h^{-1}(y;t))-\frac{\epsilon}{2},
    \end{split}
\end{equation}
where $(c)$ is a consequence of \eqref{eq:p31}, and $(d)$ follows from \ref{as:epscic}.
To summarize, we got
\begin{equation}\label{eq:p33}
    F_{U\mid d,t}(h^{-1}(y;t))-\frac{\epsilon}{2}\leq
    F_{\Y(t)\mid d}(y)
    \leq F_{U\mid d,t}(h^{-1}(y;t))+\frac{\epsilon}{2},
\end{equation}
and using \ref{as:deltacic},
\begin{equation}\label{eq:p33new}
    F_{U\mid d,\tilde{t}}(h^{-1}(y;t))-\frac{\epsilon}{2}-\delta\leq
    F_{\Y(t)\mid d}(y)
    \leq F_{U\mid d,\tilde{t}}(h^{-1}(y;t))+\frac{\epsilon}{2}+\delta,
\end{equation}
where $\tilde{t}\in\{t_0,t_1\}$.
Eq.~\eqref{eq:p33new} can also be written as
\begin{equation}\label{eq:p332}
    F_{\Y(t)\mid d}(y)-\frac{\epsilon}{2}-\delta\leq F_{U\mid d,\tilde{t}}(h^{-1}(y;t))\leq F_{\Y(t)\mid d}(y)+\frac{\epsilon}{2}+\delta,
\end{equation}
and consequently,
\begin{equation}\label{eq:p333}
    F_{U\mid d,\tilde{t}}^{-1}\big(F_{\Y(t)\mid d}(y)-\frac{\epsilon}{2}-\delta\big)\leq h^{-1}(y;t) \leq F_{U\mid d,\tilde{t}}^{-1}\big(F_{\Y(t)\mid d}(y)+\frac{\epsilon}{2}+\delta\big).
\end{equation}
Also, choosing $y=h(u;t)$ in Eq.~\eqref{eq:p33}, and applying $F^{-1}_{\Y(t)\mid d}(\cdot)$ to its sides, and renaming $t$ to $\tilde{t}$:
\begin{equation}\label{eq:p34}
    \begin{split}
        F^{-1}_{\Y(\tilde{t})\mid d}\big(F_{U\mid d,\tilde{t}}(u)-\frac{\epsilon}{2}\big)
        \leq
        h(u;\tilde{t}) \leq F^{-1}_{\Y(\tilde{t})\mid d}\big(F_{U\mid d,\tilde{t}}(u)+\frac{\epsilon}{2}\big)
    \end{split}
\end{equation}
Finally, choose $u = h^{-1}(y;t)$ in \eqref{eq:p34} and use \eqref{eq:p333}, choosing $\tilde{t}=t_0, t=t_1$ to get
\begin{equation}\label{eq:ineq2}
    F^{-1}_{\Y(t_0)\mid d}
    \big(
    F_{\Y(t_1)\mid d}(y)
    -\epsilon-\delta
    \big)
    \leq
    h\big(
    h^{-1}(y;t_1)
    ;t_0\big)
    \leq
    F^{-1}_{\Y(t_0)\mid d}
    \big(
    F_{\Y(t_1)\mid d}(y)
    +\epsilon+\delta
    \big).
\end{equation}
Note that we utilized the fact that CDFs and their inverses are increasing.
Note that Eq.~\eqref{eq:ineq2} holds both for $d=d_0,d_1$.
Combining Equations \eqref{eq:ineq2} for $d=d_0$ and $d=d_1$ yields the following two inequalities, concluding the proof after applying $F_{\Y(t_0)\vert d_1}$ to both sides.
\[
    F^{-1}_{\Y(t_0)\mid d_0}
    \big(
    F_{\Y(t_1)\mid d_0}(y)
    -\epsilon-\delta
    \big)
    \leq
    F^{-1}_{\Y(t_0)\mid d_1}
    \big(
    F_{\Y(t_1)\mid d_1}(y)
    +\epsilon+\delta
    \big),
\]
and,
\[
    F^{-1}_{\Y(t_0)\mid d_1}
    \big(
    F_{\Y(t_1)\mid d_1}(y)
    -\epsilon-\delta
    \big)
    \leq
    F^{-1}_{\Y(t_0)\mid d_0}
    \big(
    F_{\Y(t_1)\mid d_0}(y)
    +\epsilon+\delta
    \big).
\]
The result follows by \ref{as:cons}.
\end{proof}

\import{}{./appendix/appendix_asymptotic}

\section{Omitted Proofs}\label{apx:proofs}

\thmid*
\begin{proof}
\begin{equation}\label{eq:1}
    \begin{split}
        F_{\Y(t)\mid s,d}(y)
        &=\mathbbm{P}(h_{s,d}(U;t)\leq y\mid S=s,D=d)\\
        &=\mathbbm{P}(U\leq h_{s,d}^{-1}(y;t)\mid S=s,D=d)\\
        &=F_{U\mid s,d}(h_{s,d}^{-1}(y;t)),
    \end{split}
\end{equation}
and consequently,
\begin{equation}\label{eq:2}
    h_{s,d}^{-1}(y;t) = F_{U\mid s,d}^{-1}\big(F_{\Y(t)\mid s,d}(y)\big).
\end{equation}

Choose $y=h_{s,d}(u;t)$ in Eq.~\eqref{eq:1}, and apply $F^{-1}_{\Y(t)\mid s,d}(\cdot)$ to both sides:
\begin{equation}\label{eq:3}
    \begin{split}
        h_{s,d}(u;t) = F^{-1}_{\Y(t)\mid s,d}\big(F_{U\mid s,d}(u)\big)
    \end{split}
\end{equation}
Finally, combine Eq.~\eqref{eq:2} and \eqref{eq:3} to get
\begin{equation}
    h_{s,d}\big(
    h_{s,d}^{-1}(y;t)
    ;\tilde{t}\big)
    =
    F^{-1}_{\Y(\tilde{t})\mid s,d}
    \big(
    F_{\Y(t)\mid s,d}(y)
    \big).
\end{equation}

In \ref{as:stindep} (see Eq.~\ref{eq:stindep}), choose $y=h_{s_1,d_0}\big(
            h_{s_1,d_0}^{-1}(y';t_0)
            ;t_1\big)$, and take inverses
which results in
\begin{equation}
    \begin{split}
        h_{s_1,d_1}\big(
        h_{s_1,d_1}^{-1}(y;t_1)
        ;t_0\big)
        =
        h_{s_0,d_1}\Bigg(
        h^{-1}_{s_0,d_1}\bigg(
        h_{s,d}\Big\{
        h^{-1}_{s,d}\big\{
        h_{s_1,d_0}\big(
        h^{-1}_{s_1,d_0}(
        y;t_1)
        ; t_0\big)
        ; t_0\big\}
        ; t_1\Big\}
        ; t_1\bigg)
        ;t_0\Bigg),
    \end{split}
\end{equation}
or equivalently, in terms of cumulative density functions,
\begin{equation}
    \begin{split}
        F^{-1}_{\Y(t_0)\mid s_1,d_1}
        &\big(
        F_{\Y(t_1)\mid s_1,d_1}(y)
        \big)
        =\\
        &F^{-1}_{\Y(t_0)\mid s_0,d_1}
        \Bigg(
        F_{\Y(t_1)\mid s_0,d_1}\bigg(
        F^{-1}_{\Y(t_1)\mid s,d}
        \Big\{
        F_{\Y(t_0)\mid s,d}\big\{
        F^{-1}_{\Y(t_0)\mid s_1,d_0}
        \big(
        F_{\Y(t_1)\mid s_1,d_0}(
        y)
        \big)
        \big\}
        \Big\}
        \bigg)
        \Bigg),
    \end{split}
\end{equation}
which results in the identification of the density of the missing counterfactual:
\begin{multline*}
        F_{\Y(t_1)\mid s_1,d_1}(y)
        =\\
        F_{\Y(t_0)\mid s_1,d_1}
        \!\circ\!
        F^{-1}_{\Y(t_0)\mid s_0,d_1}
        \!\circ\!
        F_{\Y(t_1)\mid s_0,d_1}
        \!\circ\!
        F^{-1}_{\Y(t_1)\mid s_0,d_0}
        \!\circ\!
        F_{\Y(t_0)\mid s_0,d_0}
        \!\circ\!
        F^{-1}_{\Y(t_0)\mid s_1,d_0}
        \!\circ\!
        F_{\Y(t_1)\mid s_1,d_0}(
        y)
        .
\end{multline*}
The rest follows from \ref{as:cons}.
\end{proof}

\prppartial*
\begin{proof}
    Under \ref{as:invariance2}, \ref{as:asm} is equivalent to \ref{as:eps}.
    Further, \ref{as:invariance2} implies that \ref{as:delta} holds for $\delta=0$.
    The result follows from Proposition \ref{prp:partialgen} with $\delta=0$.
\end{proof}

\prpdelta*
\begin{proof}
    \ref{as:monotone2} implies that \ref{as:eps} holds for $\epsilon=0$.
    The result follows from Proposition \ref{prp:partialgen} with $\epsilon=0$.
\end{proof}

\prpjoint*
\begin{proof}
Under the listed assumptions, latent variable $U$ is a function of $\Y(t_0)$:
\[
\Y(t_0) = h_{s_1,d_1}(U;t_0) \Rightarrow U = h_{s_1,d_1}^{-1}(\Y(t_0);t_0),
\]
and consequently, $\Y(t_1)$ can be expressed as a function of $\Y(t_0)$:
\begin{equation}\label{eq:prp2eq1}
    \Y(t_1) = h_{s_1,d_1}(U;t_1) = h_{s_1,d_1}\big(h_{s_1,d_1}^{-1}(\Y(t_0);t_0);t_1\big).
\end{equation}
Analogous to the proof of Theorem \ref{thm:id} (see Equations \ref{eq:2} and \ref{eq:3}), the following hold for every $s,d,t$:
\[h_{s_1,d_1}^{-1}(y;t) = F_{U\mid s_1,d_1}^{-1}\big(F_{\Y(t)\mid s_1,d_1}(y)\big),\quad \text{and}\quad h_{s_1,d_1}(u;t) = F^{-1}_{\Y(t)\mid s_1,d_1}\big(F_{U\mid s_1,d_1}(u)\big),\]
and combining the two, therefore, for any $s,d$:
\begin{equation}\label{eq:prp2eq2}
    h_{s_1,d_1}\big(h_{s_1,d_1}^{-1}(y;t_0);t_1\big) = F^{-1}_{\Y(t_1)\mid s_1,d_1} \circ F_{\Y(t_0)\mid s_1,d_1}(y).
\end{equation}
From Equations \eqref{eq:prp2eq1} and \eqref{eq:prp2eq2},
\begin{equation}
    \Y(t_1) = F^{-1}_{\Y(t_1)\mid s_1,d_1} \circ F_{\Y(t_0)\mid s_1,d_1}\big(\Y(t_0)\big).
\end{equation}
Note that CDFs and their inverses are strictly increasing, and so is $\Y(t_1)$ in terms of $\Y(t_0)$.
As a result,
\[
\begin{split}
    F_{\Y(t_1), Y^1(t_1)\mid s_1,d_1}&(y^0,y^1) 
    \\&= \mathbbm{P}(\Y(t_1)\leq y^0, Y^1(t_1)\leq y^1\mid S=s_1,D=d_1)\\
    &=
    \mathbbm{P}\big(F^{-1}_{\Y(t_1)\mid s_1,d_1} \circ F_{\Y(t_0)\mid s_1,d_1}\big(\Y(t_0)\big)\leq y^0, Y^1(t_1)\leq y^1\mid S=s_1,D=d_1\big)\\
    &=
    \mathbbm{P}\big(\Y(t_0)\leq F^{-1}_{\Y(t_0)\mid s_1,d_1} \circ F_{\Y(t_1)\mid s_1,d_1}(y^0), Y^1(t_1)\leq y^1\mid S=s_1,D=d_1\big)\\
    &=
    F_{\Y(t_0), Y^1(t_1)\mid s_1,d_1}\big(F^{-1}_{\Y(t_0)\mid s_1,d_1} \circ F_{\Y(t_1)\mid s_1,d_1}(y^0),y^1\big)\\
    &=
    F_{Y(t_0), Y(t_1)\mid s_1,d_1}\big(F^{-1}_{Y(t_0)\mid s_1,d_1} \circ F_{\Y(t_1)\mid s_1,d_1}(y^0),y^1\big).
\end{split}
\]
Since the assumptions of Theorem \ref{thm:id} are satisfied, $F_{\Y(t_1)\mid s_1,d_1}(\cdot)$ is identified as in Equation \eqref{eq:thm1} of Theorem \ref{thm:id}.
\end{proof}

\thmconsistency*
\begin{proof}
\emph{Consistency.}
The estimators $\hat{F}$ converge uniformly to their corresponding $F$, and the estimators $\hat{F}^{-1}$ converge uniformly to their corresponding $F^{-1}$ (see Lemma \ref{lem:uniform}.)
As such, $\hat{F}^{-1}_{s_0d_1t_1}
        \circ
        \hat{F}_{s_0d_1t_0}
        \circ
        \hat{F}^{-1}_{s_0d_0t_0}
        \circ
        \hat{F}_{s_0d_0t_1}
        \circ
        \hat{F}^{-1}_{s_1d_0t_1}
        \circ
        \hat{F}_{s_1d_0t_0}(y) $
converges to 
$F^{-1}_{s_0d_1t_1}
        \circ
        F_{s_0d_1t_0}
        \circ
        F^{-1}_{s_0d_0t_0}
        \circ
        F_{s_0d_0t_1}
        \circ
        F^{-1}_{s_1d_0t_1}
        \circ
        F_{s_1d_0t_0}(y)$
uniformly in $y$.
It follows from the law of large numbers and \ref{as:estimation} that
$
\frac{1}{N_{s_1d_1t_0}}\sum_{i=1}^{N_{s_1d_1t_0}}\hat{F}^{-1}_{s_0d_1t_1}
        \circ
        \hat{F}_{s_0d_1t_0}
        \circ
        \hat{F}^{-1}_{s_0d_0t_0}
        \circ
        \hat{F}_{s_0d_0t_1}
        \circ
        \hat{F}^{-1}_{s_1d_0t_1}
        \circ
        \hat{F}_{s_1d_0t_0}\big(Y_{s_1d_1,i}(t_0)\big)
$
converges to $\ex{s_1d_1}\big[F^{-1}_{s_0d_1t_1}
        \circ
        F_{s_0d_1t_0}
        \circ
        F^{-1}_{s_0d_0t_0}
        \circ
        F_{s_0d_0t_1}
        \circ
        F^{-1}_{s_1d_0t_1}
        \circ
        F_{s_1d_0t_0}\big(Y(t_0)\big)\big]$.
Further, the law of large numbers also implies the convergence of $\frac{1}{N_{s_1d_1t_1}}\sum_{i=1}^{N_{s_1d_1t_1}}Y_{s_1d_1,i}(t_1)$ to $\ex{s_1d_1}\big[Y(t_1)\big]$, which completes the proof of consistency.

\emph{Asymptotic normality.}
Note that each of $\{\hat{\upmu}_i\}_{i=0}^7$ is a sample averages of i.i.d. random variables.
Under \ref{as:estimation}, the central limit theorem applies and
\begin{equation}
\begin{split}
    \sqrt{N}(\hat{\upmu}_0+\hat{\upmu}_1-\hat{\upmu}_2+\hat{\upmu}_3-\hat{\upmu}_4+\hat{\upmu}_5-\hat{\upmu}_6-&\hat{\upmu}_7)\overset{D}{\to}\mathcal{N}(0,
    V_0/p_{s_1d_1t_1}+
    V_1/p_{s_0d_1t_1}+
    V_2/p_{s_0d_1t_0}\\&+
    V_3/p_{s_0d_0t_0}+
    V_4/p_{s_0d_0t_1}+
    V_5/p_{s_1d_0t_1}+
    V_6/p_{s_1d_0t_0}+
    V_7/p_{s_1d_1t_0}
    ).
\end{split}
\end{equation}
Also, due to Lemma \ref{lem:linear} and definition of $\hat{\upmu}_0$ (Eq.~\ref{eq:muhat0}),
\begin{equation}
    \begin{split}
        \sqrt{N}\cdot\vert(\hat{\tau}-\tau)
        -(\hat{\upmu}_0+\hat{\upmu}_1-\hat{\upmu}_2+\hat{\upmu}_3-\hat{\upmu}_4+\hat{\upmu}_5-\hat{\upmu}_6-&\hat{\upmu}_7)\vert\\=
         \sqrt{N}\cdot
    \Big(
    \frac{1}{N_{s_1d_1t_0}}\sum_{i=1}^{N_{s_1d_1t_0}}
    \hat{F}^{-1}_{s_0d_1t_1}
        \circ
        \hat{F}_{s_0d_1t_0}
        \circ
        \hat{F}^{-1}_{s_0d_0t_0}
        \circ&
        \hat{F}_{s_0d_0t_1}
        \circ
        \hat{F}^{-1}_{s_1d_0t_1}
        \circ
        \hat{F}_{s_1d_0t_0}\big(
        Y_{s_1d_1,i}(t_0)
        \big) 
        \\- 
        \ex{s_1d_1}\big[F^{-1}_{s_0d_1t_1}
        \circ
        F_{s_0d_1t_0}
        \circ&
        F^{-1}_{s_0d_0t_0}
        \circ
        F_{s_0d_0t_1}
        \circ
        F^{-1}_{s_1d_0t_1}
        \circ
        F_{s_1d_0t_0}\big(
        Y(t_0)
        \big)\big]
        \\&+\hat{\upmu}_1-\hat{\upmu}_2+\hat{\upmu}_3-\hat{\upmu}_4+\hat{\upmu}_5-\hat{\upmu}_6-\hat{\upmu}_7
        \Big)=o_p(1),
    \end{split}
\end{equation}
which concludes the proof.
\end{proof}

\thmhighd*
\begin{proof}
$T^*$ is the Brenier map that pushes forward $(T_{s_0,d_0\#}\eta_{s_0,d_1})$ to $\mu_{s_0,d_1}$.
By \ref{as:stindep}, $T^*$ is also the Brenier map that pushes forward $(T_{s_1,d_0\#}\eta_{s_1,d_1})$ to $\mu_{s_1,d_1}$.
The result follows by uniqueness of these maps under the given assumptions.
\end{proof}

\section{Further on Empirical Evaluations}\label{apx:exp}
We begin with providing the complete details of the data generating mechanism we have used for our illustrations on synthetic data in the main text. 
We later provide an additional result which is postponed to this appendix.
\begin{figure}[t]
    \centering
    \includegraphics[width=.65\textwidth]{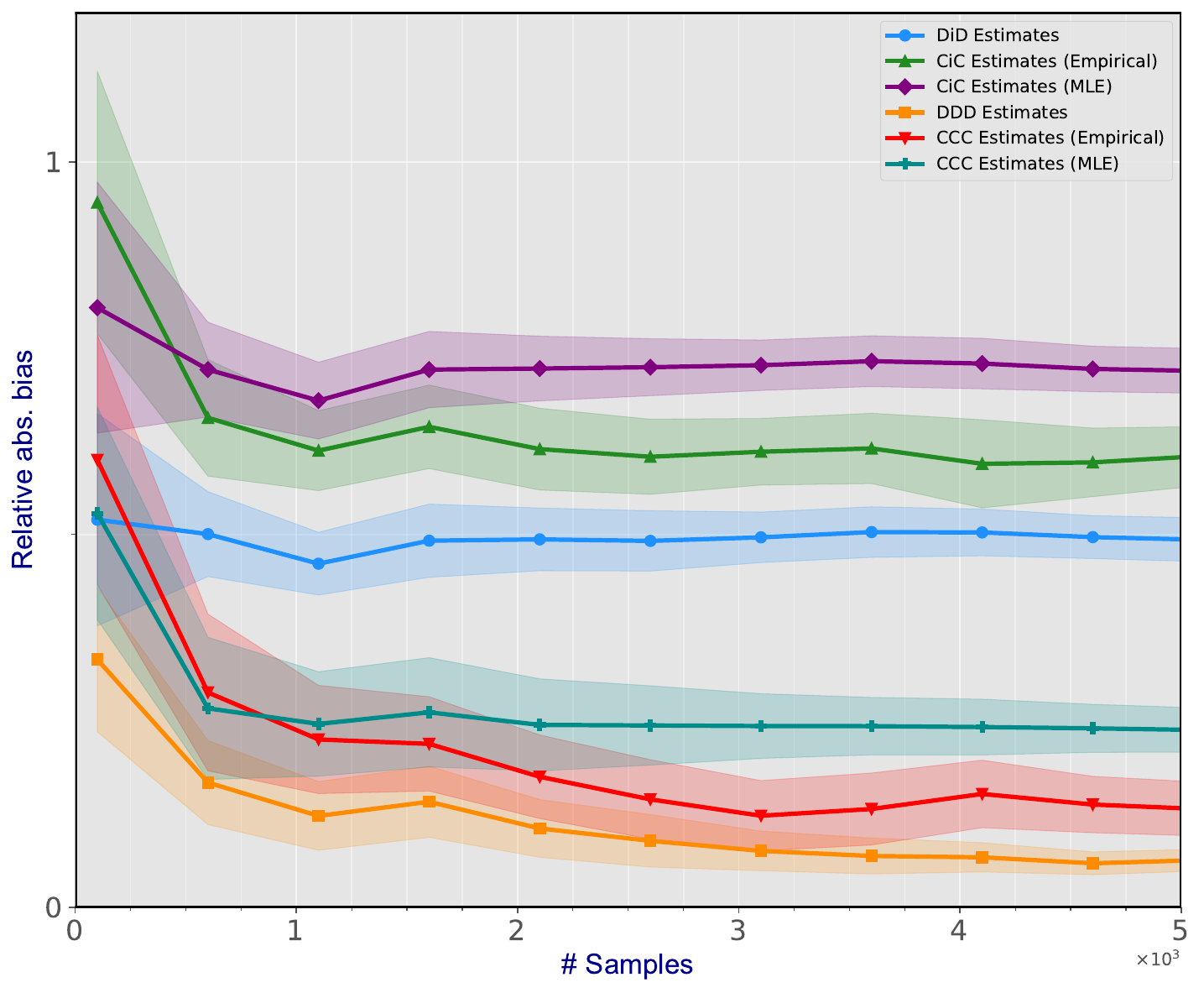}
    \caption{Relative bias of the estimators evaluated on a linear model. The models for maximum likelihood estimation are misspecified.}
    \label{fig:syn-exp}
\end{figure}
\subsection{Data generating mechanism}
To produce Figure \ref{fig:syn-linear}, we sampled the latent variable $U$ given $s,d$ from a Gaussian distribution with mean $\nu_{s,d}$ and variance $1$, where $\nu_{s,d}$s were chosen as follows:
\[
\begin{split}
    \nu_{s_0,d_0} &=0,\\
    \nu_{s_0,d_1} &=0.25, \\
    \nu_{s_1,d_0} &=-0.25,\\
    \nu_{s_1,d_1} &=0.5.
\end{split}
\]
We also generated the $\Y(t)$ counterfactuals as 
\[\Y(t)=h_{s,d}(u;t),\]
where $h_{s,d}$ was defined as a linear function:
\[h_{s,d}(u;t) \coloneqq 2u+\big(\frac{1+s}{4}+\frac{d-0.5}{2}\big) t.\]

Finally, the actual outcome $Y(t_1)$ in the group $(s_1,d_1)$ was sampled from a Gaussian distribution with mean $2.75$ and variance $1$, and the true ATT was calculated as 
$\tau = \ex{s_1d_1}[Y(t_1)] - \ex{s_1d_1}[\Y(t_1)]$.

For producing Figure \ref{fig:syn-nonlinear}, the model above was slightly modified as follows.
First, we sampled $U$ from a Gaussian distribution with the following means and variances given each group:
\[
\begin{split}
    \nu_{s_0,d_0} &=0, \quad\sigma_{s_0,d_0}=1,\\
    \nu_{s_0,d_1} &=0.25, \quad\sigma_{s_0,d_1}=1,\\
    \nu_{s_1,d_0} &=-0.25,\quad\sigma_{s_1,d_0}=1,\\
    \nu_{s_1,d_1} &=-0.5, \quad\sigma_{s_1,d_1}=1.25.
\end{split}
\]
Second, to add non-linearities to the model, two of the production functions were modified as:
\[
\begin{split}
    h_{s_0,d_1}(u;t_1) &= 0.1\exp\big(2u+\frac{1}{2} \big), \quad\text{and,}\\
    h_{s_0,d_1}(u;t_1) &= 0.1\exp\big(2u+\frac{3}{4} \big).\\
\end{split}
\]
Note that these are the same functions as the ones given in the text, after evaluating at $d=1, t=1$, and their corresponding $s\in\{0,1\}$.
Moreover, the corresponding maximum likelihood estimators were accordingly adapted to the choice of log-linear distribution instead of Gaussian.
Third, the actual outcome was sampled from a Gaussian distribution with mean $10$ and variance $1$.

\subsection{A complementary plot}
Throughout the main text, the models were correctly specified when using maximum likelihood estimators.
For the sake of completeness, herein, we provide the estimation results under a setting similar to that of Figure \ref{fig:syn-linear}, with the distinction that we sample the latent variables $U$ from an exponential distribution.
Specifically, $U$ given $s,d$ is sampled from an exponential distribution with parameters specified as follows.
\[\begin{split}
    \lambda_{s_0,d_0} = 1,\\
    \lambda_{s_0,d_1} = 2,\\
    \lambda_{s_1,d_0} = 3,\\
    \lambda_{s_1,d_1} = 1.
\end{split}\]
The outcome in the treated groups is also sampled from an exponential distribution with parameter $\frac{4}{15}$.
However, we use the Gaussian distribution for maximum likelihood estimation, causing a model mismatch.
As a result, it is expected that the CiC and CCC estimators that are based on MLE are biased due to model misspecification. 
Figure \ref{fig:syn-exp} coincides with our expectations.
Note that even the CiC estimates with MLE exhibit larger biases compared to their empirical estimator counterpart.
As such, caution is advised when using model-based estimators.
Note that since the production functions are linear, DDD is unbiased in large samples.

\end{document}

%% file: figures/fig0.tex
\begin{figure}[t]
    \centering
    \begin{tikzpicture}
        \centering\hspace{.55cm}
        \tikzstyle{block} = [circle, inner sep=1pt, fill=black]
		\tikzstyle{input} = [coordinate]
		\tikzstyle{output} = [coordinate]
        \tikzset{edge/.style = {->,> = latex',-{Latex[width=2mm]}}}
        \node[block] (y11) at (-1,1.8) {};
        \node[block] (y01) at (-1,0) {};
        \node[block] (y12) at (1,1.8) {};
        \node[block] (y02) at (1,0) {};

        \node at (-1,2+.1) {$\Y(t_0)$};
        \node at (-1,-0.2-.1) {$\Y(t_1)$};
        \node at (1,2+.1) {$\Y(t_0)$};
        \node at (1,-0.2-.1) {$\Y(t_1)$};
        
        \node[rotate=90] at (-1-.3,1.1) {$T_{0}$};
        \node[rotate=90] at (1-.3,1.1) {$T_{1}=T_0$};

        
        \draw[edge, line width=1.4pt] (y11) to (y01);
        \draw[edge, line width=1.4pt] (y12) to (y02);

        \node[rectangle, line width = 1pt, draw=blue] at (-1,-1) {$D=0$};
        \node[rectangle, line width = 1pt, draw=blue] at (1,-1) {$D=1$};
    \end{tikzpicture}
    \caption{No drift assumption of CiC framework. Mappings $T_0$ and $T_1$ are assumed to be identical.}
    \label{fig:one}
\end{figure}

%% file: figures/figure1.tex
\begin{figure*}[t]
    \centering
    \begin{tikzpicture}[scale=0.95]
        \centering\hspace{.55cm}
        \tikzstyle{block} = [circle, inner sep=1pt, fill=black]
		\tikzstyle{input} = [coordinate]
		\tikzstyle{output} = [coordinate]
        \tikzset{edge/.style = {->,> = latex',-{Latex[width=2mm]}}}
        \node[block] (y11) at (-1,2) {};
        \node[block] (y01) at (-1,0) {};
        \node[block] (y12) at (1,2) {};
        \node[block] (y02) at (1,0) {};
        \node[block] (y13) at (4+3,2) {};
        \node[block] (y03) at (4+3,0) {};
        \node[block] (y14) at (6+3,2) {};
        \node[block] (y04) at (6+3,0) {};
        \node[block] (ys01) at (3,1) {};
        \node[block] (ys02) at (8+3,1) {};

        \node at (-1,2.2+.1) {$\Y(t_0)$};
        \node at (-1,-0.2-.1) {$\Y(t_1)$};
        \node at (1,2.2+.1) {$\Y(t_0)$};
        \node at (1,-0.2-.1) {$\Y(t_1)$};
        \node at (4+3,2.2+.1) {$\Y(t_0)$};
        \node at (4+3,-0.2-.1) {$\Y(t_1)$};
        \node at (6+3,2.2+.1) {$\Y(t_0)$};
        \node at (6+3,-0.2-.1) {$\Y(t_1)$};
        \node at (3.6,1.4) {$T_{s_0,d_0}\big(\Y(t_0)\big)$};
        \node at (8.6+3,1.4) {$T_{s_1,d_0}\big(\Y(t_0)\big)$};
        
        \node[rotate=90] at (-1-.3,1.1) {$T_{s_0,d_0}$};
        \node[rotate=90] at (1-.3,1.1) {$T_{s_0,d_1}$};
        \node[rotate=90] at (4+3-.3,1.1) {$T_{s_1,d_0}$};
        \node[rotate=90] at (6+3-.3,1.1) {$T_{s_1,d_1}$};

        \node[red] at (2.4,0.15) {$T^*$};
        \node[red] at (7.4+3,0.15) {$T^*$};
        
        \draw[edge] (y11) to (y01);
        \draw[edge] (y12) to (y02);
        \draw[edge] (y13) to (y03);
        \draw[edge] (y14) to (y04);
        \draw[edge] (y12) to (ys01);
        \draw[edge] (y14) to (ys02);
        \draw[edge, red, bend left=20, line width=1.4pt] (ys01) to (y02);
        \draw[edge, red, bend left=20, line width=1.4pt] (ys02) to (y04);

        \node (up) at (5, 3.5) {};
        \node (down) at (5, -1.3) {};
        \draw[dashed, line width=.8pt] (up) to (down);
        \node[rectangle, line width = 1pt, draw=black] at (.,3.3) {$S=s_0$};
        \node[rectangle, line width = 1pt, draw=black] at (7.8,3.3) {$S=s_1$};
        
        \node[rectangle, line width = 1pt, draw=blue] at (-1,-1) {$D=d_0$};
        \node[rectangle, line width = 1pt, draw=blue] at (1.8,-1) {$D=d_1$};
        \node[rectangle, line width = 1pt, draw=blue] at (4+3,-1) {$D=d_0$};
        \node[rectangle, line width = 1pt, draw=blue] at (6.8+3,-1) {$D=d_1$};
    \end{tikzpicture}
    \caption{Illustration of the state-independent non-overlap effects assumption (\ref{as:stindep}).
    The drift mapping $T^*$ highlighted in red is identical among both states.}
    \label{fig:four}
\end{figure*}

%% file: appendix/appendix_asymptotic.tex
\section{Asymptotic Normality of $\hat{\tau}$}\label{apx:asymptotic}
Due to space limitations, we discuss the details of our estimator and its asymptotic properties in this appendix.
For simplicity of notation, we adopt $F_{sdt}$ to denote $F_{Y(t)\vert s,d}$ throughout.
We use the same convention for $F^{-1}_{Y(t)\vert s,d}$ and their corresponding estimators.
In particular, the empirical estimators were defined as
\begin{equation}\label{eq:empF}
    \hat{F}_{sdt}(y) \coloneqq N_{sdt}^{-1}\sum_{i=1}^{N_{sdt}}\ind{Y_{sd,i}(t)\leq y},
\end{equation}
and 
\begin{equation}\label{eq:empinvF}
    \hat{F}^{-1}_{sdt}(u) \coloneqq \inf\{y\in\mathbb{Y}_t, \hat{F}_{sdt}(y)\geq u\}.
\end{equation}
Our estimator for $\tau$ was also built as
\begin{multline}
    \hat{\tau}\coloneqq
    N^{-1}_{s_1d_1t_1}\sum_{i=1}^{N_{s_1d_1t_1}}Y_{s_1d_1,i}(t_1) - N^{-1}_{s_1d_1t_0}\sum_{i=1}^{N_{s_1d_1t_0}}
    \hat{F}^{-1}_{s_0d_1t_1}
        \circ
        \hat{F}_{s_0d_1t_0}
        \circ
        \hat{F}^{-1}_{s_0d_0t_0}
        \circ
        \hat{F}_{s_0d_0t_1}
        \circ
        \hat{F}^{-1}_{s_1d_0t_1}
        \circ
        \hat{F}_{s_1d_0t_0}\big(
        Y_{s_1d_1,i}(t_0)
        \big).
\end{multline}
For ease of notation, in what follows, we adopt the notation $\ex{sd}$ to denote expectation conditioned on $S=s,D=d$.
We define the following functions:
\begin{equation}\label{eq:gs}
\begin{split}
    g_2(\cdot)\coloneqq F^{-1}_{s_0d_1t_1}(\cdot),\quad
    &g_3(\cdot)\coloneqq g_2
        \circ
        F_{s_0d_1t_0}(\cdot),\quad
    g_4(\cdot)\coloneqq g_3
        \circ
        F^{-1}_{s_0d_0t_0}(\cdot),\\
    g_5(\cdot)\coloneqq &\:g_4
        \circ
        F_{s_0d_0t_1}(\cdot),\quad
    g_6(\cdot)\coloneqq g_5
        \circ
        F^{-1}_{s_1d_0t_1}(\cdot),\quad
    g_7(\cdot)\coloneqq g_6
        \circ
        F_{s_1d_0t_0}(\cdot),
\end{split}
\end{equation}
as well as
\begin{equation}\label{eq:rs}
    \begin{split}
        r_5(\cdot)\coloneqq F_{s_1d_0t_0}(\cdot),\quad
        &r_4(\cdot)\coloneqq
        F^{-1}_{s_1d_0t_1}
        \circ
        r_5(\cdot),\quad
        r_3(\cdot)\coloneqq F_{s_0d_0t_1}
        \circ
        r_4(\cdot),\\
        r_2(\cdot)\coloneqq &\:F^{-1}_{s_0d_0t_0}
        \circ
        r_3(\cdot),\quad
        r_1(\cdot)\coloneqq F_{s_0d_1t_0}
        \circ
        r_2(
        \cdot
        ).
    \end{split}
\end{equation}

Accordingly, define $Q_0(y) = y-\ex{s_1d_1}[Y(t_1)]$, and
\begin{equation}
\begin{split}
    q_1(y,z)&\coloneqq \frac{
    \ind{F_{s_0d_1t_1}(y)\leq
    r_1(z)
    }
    -r_1(z)
    }{f_{s_0d_1t_1}\circ
    F^{-1}_{s_0d_1t_1}
    \circ r_1(z)
    },\quad Q_1(y)\coloneqq\ex{s_1d_1}[q_1(y,Y(t_0))]\\
    q_2(y,z) &\coloneqq
    \big(\ind{y\leq
    r_2(z)
    }
    -
    F_{s_0d_1t_0}
    \circ
    r_2(z)
    \big)\cdot
    \big(g_2'\circ F_{s_0d_1t_0}
    \circ
    r_2(z)
    \big),\quad Q_2(y)\coloneqq\ex{s_1d_1}[q_2(y,Y(t_0))]\\
    q_3(y,z)&\coloneqq
    \frac{
    \ind{F_{s_0d_0t_0}(y)\leq
    r_3(z) 
    }
    -r_3(z)
    }{f_{s_0d_0t_0}\circ
    F^{-1}_{s_0d_0t_0}\circ
    r_3(z) }
    \cdot\big(g_3'\circ f_{s_0d_0t_0}\circ F^{-1}_{s_0d_0t_0}
    \circ
    r_3(z)
    \big),\quad Q_3(y)\coloneqq\ex{s_1d_1}[q_3(y,Y(t_0))]\\
    q_4(y,z)&\coloneqq 
    \big(\ind{y\leq
    r_4\big(z\big)  
    }
    -
    F_{s_0d_0t_1}
    \circ
    r_4(z)  
    \big)\cdot
    \big(g_4'\circ F_{s_0d_0t_1}
    \circ
    r_4(z)\big) 
    ,\quad Q_4(y)\coloneqq\ex{s_1d_1}[q_4(y,Y(t_0))]\\
    q_5(y,z)&\coloneqq
    \frac{
    \ind{F_{s_1d_0t_1}(y)\leq
    r_5(z) 
    }
    -r_5(z) 
    }{f_{s_1d_0t_1}\circ
    F^{-1}_{s_1d_0t_1}\circ
    r_5(z) }\cdot
    \big(g_5'\circ f_{s_1d_0t_1}\circ F^{-1}_{s_1d_0t_1}
    \circ
    r_5(z)
    \big)\quad Q_5(y)\coloneqq\ex{s_1d_1}[q_5(y,Y(t_0))]
    ,\\
    q_6(y,z) &\coloneqq
    \big(\ind{y\leq
    z 
    }
    -
    F_{s_1d_0t_0}(z)  
    \big)\cdot
    \big(g_6'\circ F_{s_1d_0t_0}(z) 
    \big),\quad Q_6(y)\coloneqq\ex{s_1d_1}[q_6(y,Y(t_0))],\\
    Q_7(y)&\coloneqq
    g_7(y)-
    \ex{s_1d_1}\big[g_7\big(
        Y(t_0)
        \big)\big].
\end{split}
\end{equation}

Also define \begin{equation}\label{eq:muhat0}\hat{\upmu}_0\coloneqq\frac{1}{N_{s_1d_1t_1}}\sum_{i=1}^{N_{s_1d_1t_1}}Q_0\big(Y_{s_1d_1,i}(t_1)\big),\end{equation} 
and,

\begin{equation}\label{eq:muhat1}
    \hat{\mu}_1\coloneqq\frac{1}{N_{s_1d_1t_0}}\frac{1}{N_{s_0d_1t_1}}\sum_{i=1}^{N_{s_1d_1t_0}}\sum_{j=1}^{N_{s_0d_1t_1}}
    q_1(
    Y_{s_0d_1t,j}(t_1)
    ,Y_{s_1d_1,i}(t_0)
    ),\quad\quad\hat{\upmu}_1 \coloneqq
    \frac{1}{N_{s_0d_1t_1}}\sum_{j=1}^{N_{s_0d_1t_1}}
    Q_1(
    Y_{s_0d_1t,j}(t_1)
    ),
\end{equation}

\begin{equation}\label{eq:muhat2}
    \hat{\mu}_2=
    \frac{1}{N_{s_1d_1t_0}}\frac{1}{N_{s_0d_1t_0}}\sum_{i=1}^{N_{s_1d_1t_0}}\sum_{j=1}^{N_{s_0d_1t_0}}
    q_2(
    Y_{s_0d_1t,j}(t_0)
    ,Y_{s_1d_1,i}(t_0)
    ),\quad\quad
    \hat{\upmu}_2=
    \frac{1}{N_{s_0d_1t_0}}\sum_{j=1}^{N_{s_0d_1t_0}}
    Q_2(
    Y_{s_0d_1t,j}(t_0)
    ),
\end{equation}

\begin{equation}\label{eq:muhat3}
    \hat{\mu}_3=
    \frac{1}{N_{s_1d_1t_0}}\frac{1}{N_{s_0d_0t_0}}\sum_{i=1}^{N_{s_1d_1t_0}}\sum_{j=1}^{N_{s_0d_0t_0}}
    q_3(
    Y_{s_0d_0t,j}(t_0)
    ,Y_{s_1d_1,i}(t_0)
    ),\quad\quad
    \hat{\upmu}_3=
    \frac{1}{N_{s_0d_0t_0}}\sum_{j=1}^{N_{s_0d_0t_0}}
    Q_3(
    Y_{s_0d_0t,j}(t_0)
    ),
\end{equation}

\begin{equation}\label{eq:muhat4}
    \hat{\mu}_4=
    \frac{1}{N_{s_1d_1t_0}}\frac{1}{N_{s_0d_0t_1}}\sum_{i=1}^{N_{s_1d_1t_0}}\sum_{j=1}^{N_{s_0d_0t_1}}
    q_4(
    Y_{s_0d_0t,j}(t_1)
    ,Y_{s_1d_1,i}(t_0)
    ),\quad\quad
    \hat{\upmu}_4=
    \frac{1}{N_{s_0d_0t_1}}\sum_{j=1}^{N_{s_0d_0t_1}}
    Q_4(
    Y_{s_0d_0t,j}(t_1)
    ),
\end{equation}
\begin{equation}\label{eq:muhat5}
    \hat{\mu}_5=
    \frac{1}{N_{s_1d_1t_0}}\frac{1}{N_{s_1d_0t_1}}\sum_{i=1}^{N_{s_1d_1t_0}}\sum_{j=1}^{N_{s_1d_0t_1}}
    q_5(
    Y_{s_0d_1t,j}(t_1)
    ,Y_{s_1d_1,i}(t_0)
    ),\quad\quad
    \hat{\upmu}_5=
    \frac{1}{N_{s_1d_0t_1}}\sum_{j=1}^{N_{s_1d_0t_1}}
    Q_5(
    Y_{s_0d_1t,j}(t_1)
    ),
\end{equation}

\begin{equation}\label{eq:muhat6}
    \hat{\mu}_6=
    \frac{1}{N_{s_1d_1t_0}}\frac{1}{N_{s_1d_0t_0}}\sum_{i=1}^{N_{s_1d_1t_0}}\sum_{j=1}^{N_{s_1d_0t_0}}
    q_6(
    Y_{s_0d_1t,j}(t_0)
    ,Y_{s_1d_1,i}(t_0)
    ),\quad\quad
    \hat{\upmu}_6=
    \frac{1}{N_{s_1d_0t_0}}\sum_{j=1}^{N_{s_1d_0t_0}}
    Q_6(
    Y_{s_0d_1t,j}(t_0)
    ),
\end{equation}
and,
\begin{equation}\label{eq:muhat7}
\begin{split}
    \hat{\upmu}_7 \coloneqq
    \frac{1}{N_{s_1d_1t_0}}\sum_{i=1}^{N_{s_1d_1t_0}}
    Q_7\big(
        Y_{s_1d_1,i}(t_0)\big).
\end{split}
\end{equation}
Finally, define the variance terms $\{V_0,\dots, V_7\}$ as
\begin{equation}\label{eq:varianceterms}
    \begin{split}
        V_0&\coloneqq \ex{s_1d_1}[Q_0\big(Y(t_1)\big)^2],\\
        V_1&\coloneqq \ex{s_0d_1}[Q_1\big(Y(t_0)\big)^2],\\
        V_2&\coloneqq \ex{s_0d_1}[Q_2\big(Y(t_1)\big)^2],\\
        V_3&\coloneqq \ex{s_0d_0}[Q_3\big(Y(t_1)\big)^2],\\
        V_4&\coloneqq \ex{s_0d_0}[Q_4\big(Y(t_0)\big)^2],\\
        V_5&\coloneqq \ex{s_0d_1}[Q_5\big(Y(t_0)\big)^2],\\
        V_6&\coloneqq \ex{s_0d_1}[Q_6\big(Y(t_1)\big)^2],\\
        V_7&\coloneqq \ex{s_1d_1}[Q_7\big(Y(t_0)\big)^2].\\
    \end{split}
\end{equation}
For the sake of completeness, we review the consistency and asymptotic normality result from the main text for $\hat{\tau}$:
\thmconsistency*

In order to estimate the asymptotic variance, one can substitute expectations with sample averages. This involves utilizing empirical distributions and their inverses for cumulative distribution functions, as well as employing any uniformly consistent non-parametric estimator for the density functions.

Twelve lemmas are in order within the next section which will be used to prove Theorem \ref{thm:consistency} in Section \ref{apx:proofs}.
\section{Preliminary Lemmas}\label{apx:lemmas}
Throughout this subsection, we suppress the subscripts $sdt$ and replace $N_{sdt}$ by $N$ whenever it does not affect the statements.
Note that since $N_{sdt}/N\to p_{sdt}>0$ (see \ref{as:estimation}), if a term is $O_p(N_{sdt}^{-\alpha})$ ($o_p(N_{sdt}^{-\alpha})$), it is also $O_p(N^{-\alpha})$ ($o_p(N^{-\alpha})$).
For the following four lemmas, let $X$ be a random variable with twice continuously differentiable cdf $F(\cdot)$, where $F$ has bounded first and second derivatives. Suppose the pdf of $X$ is bounded away from $0$.
Let $\hat{F}$ and $\hat{F}^{-1}$ be empirical estimators of the CDF of $X$ and its inverse, analogous to Eq.~\eqref{eq:empF} and Eq.~\eqref{eq:empinvF}.

\begin{lemma}[\citealp{athey2006identification}]\label{lem:uniform}
    For any $\delta<1/2$,
    \[
        \sup_x N^\delta\cdot\big\vert\hat{F}(x)-F(x)\big\vert\overset{P}{\to}0, \quad\quad \sup_{u\in[0,1]} N^\delta\cdot\big\vert\hat{F}^{-1}(u)-F^{-1}(u)\big\vert\overset{P}{\to}0.
    \]
\end{lemma}

The following lemma is a slight generalization of the latter.
\begin{lemma}\label{lem:guniform}
    Let $g(\cdot)$ be a continuously differentiable function with a bounded derivative. 
    For any $\delta< \frac{1}{2}$, 
    \[
        \sup_x N^\delta\cdot\big\vert g\circ\hat{F}(x)- g\circ F(x)\big\vert\overset{P}{\to}0, \quad\quad \sup_{u\in[0,1]} N^\delta\cdot\big\vert g\circ \hat{F}^{-1}(u)- g\circ F^{-1}(u)\big\vert\overset{P}{\to}0.
    \]
\end{lemma}
\begin{proof}
    \[
        \sup_xN^\delta\cdot\big\vert g\circ\hat{F}(x)- g\circ F(x)\big\vert
        \leq \sup_{x,\tilde{x}}N^\delta\cdot\big\vert g'(\tilde{x})\big\vert\cdot\big\vert\hat{F}(x)- F(x)\big\vert
        \leq \sup_{x,\tilde{x}}\big\vert g'(\tilde{x})\big\vert\cdot\sup_xN^\delta\cdot\big\vert\hat{F}(x)- F(x)\big\vert,
    \]
    and the result follows from Lemma \ref{lem:uniform}.
    The proof for the second part is identical.
\end{proof}
\begin{lemma}\label{lem:genuniform}
    Let $\{X_i\}_{i=1}^k$ be random variables with twice differentiable cdfs $\{F_i(\cdot)\}_{i=1}^k$, where $F_i$s have bounded first and second derivatives.
    Suppose also that all $X_i$s have pdfs that are bounded away from $0$.
    Then for any $\delta< \frac{1}{2}$, 
    \[
    \sup_xN^\delta\cdot\big\vert \hat{g}_1\circ\dots\circ\hat{g}_k(x)-g_1\circ\dots\circ g_k(x)\big\vert\overset{P}{\to}0,
    \]
    where $g_i\in\{F_i,F^{-1}_i\}$, and $\hat{g}_i$ is the empirical estimator of $g_i$, analogous to Equations \eqref{eq:empF} and \eqref{eq:empinvF}.
\end{lemma}
\begin{proof}
    Using triangle inequality, 
    \[\begin{split}
        \sup_xN^\delta&\cdot\big\vert \hat{g}_1\circ\dots\circ\hat{g}_k(x)-g_1\circ\dots\circ g_k(x)\big\vert\leq
        \\
        &\sup_xN^\delta\cdot\big\vert \hat{g}_1\circ\dots\circ\hat{g}_k(x)-g_1\circ\hat{g_2}\circ\dots\circ \hat{g}_k(x)\big\vert
        \\
        +&\sup_xN^\delta\cdot\big\vert g_1\circ\hat{g_2}\circ\dots\circ \hat{g}_k(x)-g_1\circ g_2(x)\circ\hat{g_3}\circ\dots\circ \hat{g}_k(x)\big\vert
        \\
        +&\dots
        \\
        +&
        \sup_xN^\delta\cdot\big\vert g_1\circ\dots\circ g_{k-1}\circ\hat{g}_k(x)-g_1\circ\dots\circ g_k(x)\big\vert,
    \end{split}\]
    and every term on the right hand side is $o_p(1)$ due to Lemma \ref{lem:guniform}.
\end{proof}

\begin{lemma}\label{lem:gena6}
    Let $g(\cdot)$ be a twice continuously differentiable function, with bounded first and second derivatives $g'(\cdot)$ and $g''(\cdot)$, respectively.
    Then for all $0<\eta<5/7$,
    \begin{equation}
        \sup_{u\in[0,1]} N^{\eta}\cdot\Big\vert
        g\circ\hat{F}^{-1}(u)-
        g\circ F^{-1}(u)
        +\frac{g'\circ f\circ F^{-1}(u)}{f\circ F^{-1}(u)}\big(
        \hat{F}\circ F^{-1}(u)-u
        \big)
        \Big\vert \overset{P}{\to}0.
    \end{equation}
\end{lemma}
\begin{proof}
    Using triangle inequality,
    \begin{multline}\label{eq:proofgena6}
        \sup_{u\in[0,1]} N^{\eta}\cdot\Big\vert
        g\circ\hat{F}^{-1}(u)-
        g\circ F^{-1}(u)
        +\frac{g'\circ f\circ F^{-1}(u)}{f\circ F^{-1}(u)}\big(
        \hat{F}\circ F^{-1}(u)-u
        \big)
        \Big\vert
        \\
        \leq
        \sup_{u\in[0,1]} N^{\eta}\cdot\Big\vert
        g\circ\hat{F}^{-1}(u)-
        g\circ F^{-1}\circ
        \hat{F}\circ\hat{F}^{-1}
        (u)
        +
        \frac{g'\circ f\circ \hat{F}^{-1}(u)}{f\circ \hat{F}^{-1}(u)}\big(
        \hat{F}\circ \hat{F}^{-1}(u)-
        F\circ \hat{F}^{-1}(u)
        \big)
        \Big\vert
        \\
        +
        \sup_{u\in[0,1]} N^{\eta}\cdot\Big\vert
        \frac{g'\circ f\circ F^{-1}(u)}{f\circ F^{-1}(u)}\big(
        \hat{F}\circ F^{-1}(u)-u
        \big)
        -
        \frac{g'\circ f\circ \hat{F}^{-1}(u)}{f\circ \hat{F}^{-1}(u)}\big(
        \hat{F}\circ \hat{F}^{-1}(u)-
        F\circ \hat{F}^{-1}(u)
        \big)
        \Big\vert
        \\
        +
        \sup_{u\in[0,1]} N^{\eta}\cdot\Big\vert
        g\circ F^{-1}\circ
        \hat{F}\circ\hat{F}^{-1}
        (u) - g\circ F^{-1}(u)
        \Big\vert,
    \end{multline}
where we shall show each of the three terms on the right hand side are $o_p(1)$.
The first term is bounded as 
\begin{multline}
    \sup_x N^{\eta}\cdot\Big\vert
        g(x)-
        g\circ F^{-1}\circ
        \hat{F}
        (x)
        +
        \frac{g'\circ f(x)}{f(x)}\big(
        \hat{F}(x)-
        F(x)
        \big)
        \Big\vert
    \\=
    \sup_x N^{\eta}\cdot\Big\vert
        g\circ F^{-1}\big( F(x)\big)-
        g\circ F^{-1}\big(
        \hat{F}
        (x)\big)
        -
        \frac{g'\circ f(x)}{f\circ F^{-1}\big(F(x)\big)}\big(
        F(x)-
        \hat{F}(x)
        \big)
        \Big\vert
    \\\leq
    \sup_{\tilde{x},x}N^\eta\cdot\Big\vert
    \big(
    -\frac{1}{f(\tilde{x})^3}\frac{\partial f}{\partial x}(\tilde{x})g'\circ F^{-1}(\tilde{x})
    +\frac{1}{f(\tilde{x})^2}g''\circ F^{-1}(\tilde{x})
    \big)
    \big(F(x)-
        \hat{F}(x)\big)^2
    \Big\vert,
\end{multline}
where we used the expansion of the function $g\circ F^{-1}(\cdot)$ around the point $F(x)$.
The latter is $o_p(1)$ for every $\eta<1$ due to Lemma \ref{lem:uniform} and smoothness assumptions over $g$ and $f$.
Next, the second term of Eq.~\eqref{eq:proofgena6} can be bounded using triangle inequality by
\begin{multline}
    \sup_{u\in[0,1]} N^{\eta}\cdot\Big\vert
        \frac{g'\circ f\circ F^{-1}(u)}{f\circ F^{-1}(u)}\big(
        \hat{F}\circ F^{-1}(u)-u
        \big)
        -
        \frac{g'\circ f\circ \hat{F}^{-1}(u)}{f\circ \hat{F}^{-1}(u)}\big(
        \hat{F}\circ \hat{F}^{-1}(u)-
        F\circ \hat{F}^{-1}(u)
        \big)
        \Big\vert
        \\\leq
    \sup_{u\in[0,1]} N^{\eta}\cdot\Big\vert
        \frac{g'\circ f\circ F^{-1}(u)}{f\circ F^{-1}(u)}\big(
        \hat{F}\circ F^{-1}(u)-u
        \big)
        -
        \frac{g'\circ f\circ \hat{F}^{-1}(u)}{f\circ \hat{F}^{-1}(u)}\big(
        \hat{F}\circ F^{-1}(u)-
        u
        \big)
        \Big\vert
        \\+
    \sup_{u\in[0,1]} N^{\eta}\cdot\Big\vert
        \frac{g'\circ f\circ \hat{F}^{-1}(u)}{f\circ \hat{F}^{-1}(u)}\big(
        \hat{F}\circ F^{-1}(u)-
        u
        \big)
        -
        \frac{g'\circ f\circ \hat{F}^{-1}(u)}{f\circ \hat{F}^{-1}(u)}\big(
        \hat{F}\circ \hat{F}^{-1}(u)-
        F\circ \hat{F}^{-1}(u)
        \big)
        \Big\vert
        \\\leq
        \sup_{u\in[0,1]} N^{\eta/2}\cdot\Big\vert
        \hat{F}\circ F^{-1}(u)-u
        \Big\vert\cdot
        \sup_{u\in[0,1]} N^{\eta/2}\cdot\Big\vert
        \frac{g'\circ f\circ F^{-1}(u)}{f\circ F^{-1}(u)}
        -
        \frac{g'\circ f\circ \hat{F}^{-1}(u)}{f\circ \hat{F}^{-1}(u)}
        \Big\vert
        \\+
        \sup_{x}\Big\vert
        \frac{g'\circ f(x)}{f(x)}
        \Big\vert\cdot
        \sup_{u\in[0,1]} N^{\eta}\big\vert
        \hat{F}\circ F^{-1}(u)-
        u
        -
        \big(
        \hat{F}\circ \hat{F}^{-1}(u)-
        F\circ \hat{F}^{-1}(u)
        \big)
        \big\vert,
\end{multline}
where the first term is $o_p(1)$ due to Lemma \ref{lem:uniform}, Lemma A.3 of \citet{athey2006identification}, and the smoothness assumptions over $f$ and $g$,
and the second term is $o_p(1)$ as shown in the proof of Lemma A.6 of \citet{athey2006identification}.
Finally, consider the third term of Eq.~\eqref{eq:proofgena6}:
\begin{equation}
    \sup_{u\in[0,1]} N^{\eta}\cdot\Big\vert
        g\circ F^{-1}\circ
        \hat{F}\circ\hat{F}^{-1}
        (u) - g\circ F^{-1}(u)
        \Big\vert.
\end{equation}
Since $\vert\hat{F}\circ\hat{F}^{-1}(u)-u\vert< 1/N$ for all $u$ \cite{athey2006identification}, and due to smoothness assumptions over $g\circ F^{-1}$, this term converges uniformly to zero.
\end{proof}

\begin{lemma}\label{lem:muupmu}
    Let $\{\hat{\mu}_k\}_{k=1}^6$ and $\{\hat{\upmu}_k\}_{k=1}^6$ be given by Equations \eqref{eq:muhat1} through \eqref{eq:muhat6}.
    Under \ref{as:estimation}, for any $k\in\{1,\dots,6\}$,
    \[\sqrt{N}(\hat{\mu}_k-\hat{\upmu}_k)\overset{P}{\to}0.\]
\end{lemma}
\begin{proof}
    We give the proof for $k=1$. 
    The proof is identical for the rest.
    $\hat{\mu}_1$ is a twp-sample V-statistic, and standard V-statistic theory implies
    \[\begin{split}
        \hat{\mu}_1 &= \frac{1}{N_{s_0d_1t_1}}\sum_{i=1}^{N_{s_0d_1t_1}}\ex{s_1d_1}\big[q_1\big(Y_{s_0d_1,i}(t_1), Y(t_0)\big)\big] 
    + \frac{1}{N_{s_1d_1t_0}}\sum_{i=1}^{N_{s_1d_1t_0}}\ex{s_0d_1}\big[q_1\big(Y(t_0), Y_{s_1d_1,i}(t_0)\big)\big] 
    + o_p(N^{-1/2})\\
    &= \frac{1}{N_{s_0d_1t_1}}\sum_{i=1}^{N_{s_0d_1t_0}}Q_1\big(Y_{s_0d_1,i}(t_1)\big)
    + 0 
    + o_p(N^{-1/2})=\hat{\upmu}_1+o_p(N^{-1/2}),
    \end{split}
    \]
    which completes the proof.
\end{proof}

\import{./appendix/}{third_term}

\begin{lemma}\label{lem:term1} Define $\hat{\mu}_1$ as given by Eq.~\eqref{eq:muhat1}. Under \ref{as:estimation},
    \begin{multline}\label{eq:lemterm1}
        \sqrt{N}\cdot
        \Big\vert
        \frac{1}{N_{s_1d_1t_0}}\sum_{i=1}^{N_{s_1d_1t_0}}
        \hat{F}^{-1}_{s_0d_1t_1}
        \circ
        \hat{F}_{s_0d_1t_0}
        \circ
        \hat{F}^{-1}_{s_0d_0t_0}
        \circ
        \hat{F}_{s_0d_0t_1}
        \circ
        \hat{F}^{-1}_{s_1d_0t_1}
        \circ
        \hat{F}_{s_1d_0t_0}\big(
        Y_{s_1d_1,i}(t_0)
        \big) 
        \\- 
        \frac{1}{N_{s_1d_1t_0}}\sum_{i=1}^{N_{s_1d_1t_0}}
        F^{-1}_{s_0d_1t_1}
        \circ
        \hat{F}_{s_0d_1t_0}
        \circ
        \hat{F}^{-1}_{s_0d_0t_0}
        \circ
        \hat{F}_{s_0d_0t_1}
        \circ
        \hat{F}^{-1}_{s_1d_0t_1}
        \circ
        \hat{F}_{s_1d_0t_0}\big(
        Y_{s_1d_1,i}(t_0)
        \big) 
        +\hat{\mu}_1
        \Big\vert
    \overset{P}{\to}0.\end{multline}
\end{lemma}
\begin{proof}
    The proof is analogous to that of Lemma \ref{lem:term3}.
    To adapt the proof, we just need to replace the following definitions:
    Define
        $\hat{q}_i \coloneqq 
        \hat{F}^{-1}_{s_0d_1t_0}
        \circ
        \hat{F}^{-1}_{s_0d_0t_0}
        \circ
        \hat{F}_{s_0d_0t_1}
        \circ
        \hat{F}^{-1}_{s_1d_0t_1}
        \circ
        \hat{F}_{s_1d_0t_0}\big(
        Y_{s_1d_1,i}(t_0)
        \big) $, 
        $q_i \coloneqq 
        F^{-1}_{s_0d_1t_0}
        \circ
        F^{-1}_{s_0d_0t_0}
        \circ
        F_{s_0d_0t_1}
        \circ
        F^{-1}_{s_1d_0t_1}
        \circ
        F_{s_1d_0t_0}\big(
        Y_{s_1d_1,i}(t_0)
        \big)$, and $g(\cdot)$ is the identity function.
        Finally, the subscripts $s_0d_0t_0$ are replaced by $s_0d_1t_1$.
        The rest of the proof is identical.
\end{proof}

\begin{lemma}\label{lem:term5} Define $\hat{\mu}_5$ as given by Eq.~\eqref{eq:muhat5}. Under \ref{as:estimation},
    \begin{multline}\label{eq:lemterm5}
        \sqrt{N}\cdot
        \Big\vert
        \frac{1}{N_{s_1d_1t_0}}\sum_{i=1}^{N_{s_1d_1t_0}}
        F^{-1}_{s_0d_1t_1}
        \circ
        F_{s_0d_1t_0}
        \circ
        F^{-1}_{s_0d_0t_0}
        \circ
        F_{s_0d_0t_1}
        \circ
        \hat{F}^{-1}_{s_1d_0t_1}
        \circ
        \hat{F}_{s_1d_0t_0}\big(
        Y_{s_1d_1,i}(t_0)
        \big) 
        \\- 
        \frac{1}{N_{s_1d_1t_0}}\sum_{i=1}^{N_{s_1d_1t_0}}
        F^{-1}_{s_0d_1t_1}
        \circ
        F_{s_0d_1t_0}
        \circ
        F^{-1}_{s_0d_0t_0}
        \circ
        F_{s_0d_0t_1}
        \circ
        F^{-1}_{s_1d_0t_1}
        \circ
        \hat{F}_{s_1d_0t_0}\big(
        Y_{s_1d_1,i}(t_0)
        \big) 
        +\hat{\mu}_5
        \Big\vert
    \overset{P}{\to}0.\end{multline}
\end{lemma}
\begin{proof}
    The proof is analogous to that of Lemma \ref{lem:term3}.
    To adapt the proof, we just need to replace the following definitions:
    Define
        $\hat{q}_i \coloneqq 
        \hat{F}_{s_1d_0t_0}\big(
        Y_{s_1d_1,i}(t_0)
        \big) $, 
        $q_i \coloneqq 
        F_{s_1d_0t_0}\big(
        Y_{s_1d_1,i}(t_0)
        \big)$, and $g(\cdot)\coloneqq F^{-1}_{s_0d_1t_1}
        \circ
        F_{s_0d_1t_0}
        \circ
        F^{-1}_{s_0d_0t_0}
        \circ
        F_{s_0d_0t_1}(\cdot)$.
        Finally, the subscripts $s_0d_0t_0$ are replaced by $s_1d_0t_1$.
        The rest of the proof is identical.
\end{proof}

\import{}{./appendix/fourth_term}

\begin{lemma}\label{lem:linear}
    Let $\{\hat{\upmu}\}_{i=1}^7$ be given by Equations \eqref{eq:muhat1} through \eqref{eq:muhat7}.
    Under \ref{as:estimation},
    \begin{multline}\label{eq:lemlinear}
        \sqrt{N}\cdot
    \Big(
    \frac{1}{N_{s_1d_1t_0}}\sum_{i=1}^{N_{s_1d_1t_0}}
    \hat{F}^{-1}_{s_0d_1t_1}
        \circ
        \hat{F}_{s_0d_1t_0}
        \circ
        \hat{F}^{-1}_{s_0d_0t_0}
        \circ
        \hat{F}_{s_0d_0t_1}
        \circ
        \hat{F}^{-1}_{s_1d_0t_1}
        \circ
        \hat{F}_{s_1d_0t_0}\big(
        Y_{s_1d_1,i}(t_0)
        \big) 
        \\- 
        \ex{s_1d_1}\big[F^{-1}_{s_0d_1t_1}
        \circ
        F_{s_0d_1t_0}
        \circ
        F^{-1}_{s_0d_0t_0}
        \circ
        F_{s_0d_0t_1}
        \circ
        F^{-1}_{s_1d_0t_1}
        \circ
        F_{s_1d_0t_0}\big(
        Y(t_0)
        \big)\big]
        \\+\hat{\upmu}_1-\hat{\upmu}_2+\hat{\upmu}_3-\hat{\upmu}_4+\hat{\upmu}_5-\hat{\upmu}_6-\hat{\upmu}_7
        \Big)\overset{P}{\to}0.
    \end{multline}
\end{lemma}
\begin{proof}
    Using triangle inequality, the absolute value of the left hand side of Eq.~\eqref{eq:lemlinear} is bounded by
    \begin{multline}\label{eq:prooflinear}
        \sqrt{N}\cdot
    \Big\vert
    \frac{1}{N_{s_1d_1t_0}}\sum_{i=1}^{N_{s_1d_1t_0}}
    \hat{F}^{-1}_{s_0d_1t_1}
        \circ
        \hat{F}_{s_0d_1t_0}
        \circ
        \hat{F}^{-1}_{s_0d_0t_0}
        \circ
        \hat{F}_{s_0d_0t_1}
        \circ
        \hat{F}^{-1}_{s_1d_0t_1}
        \circ
        \hat{F}_{s_1d_0t_0}\big(
        Y_{s_1d_1,i}(t_0)
        \big) 
        \\- 
        \ex{s_1d_1}\big[F^{-1}_{s_0d_1t_1}
        \circ
        F_{s_0d_1t_0}
        \circ
        F^{-1}_{s_0d_0t_0}
        \circ
        F_{s_0d_0t_1}
        \circ
        F^{-1}_{s_1d_0t_1}
        \circ
        F_{s_1d_0t_0}\big(
        Y(t_0)
        \big)\big]
        \\+\hat{\mu}_1-\hat{\mu}_2+\hat{\mu}_3-\hat{\mu}_4+\hat{\mu}_5-\hat{\mu}_6-\hat{\upmu}_7
        \Big\vert
        +\sqrt{N}\cdot\sum_{i=1}^6\vert\hat{\mu}_i-\hat{\upmu}_i\vert.
    \end{multline}
The term $\sqrt{N}\cdot\sum_{i=1}^6\vert\hat{\mu}_i-\hat{\upmu}_i\vert$ is $o_p(1)$ due to Lemma \ref{lem:muupmu}.
The first term of Eq.~\eqref{eq:prooflinear} is bounded by
\end{proof}

\begin{multline}
        \sqrt{N}\cdot
        \Big\vert
        \frac{1}{N_{s_1d_1t_0}}\sum_{i=1}^{N_{s_1d_1t_0}}
        \hat{F}^{-1}_{s_0d_1t_1}
        \circ
        \hat{F}_{s_0d_1t_0}
        \circ
        \hat{F}^{-1}_{s_0d_0t_0}
        \circ
        \hat{F}_{s_0d_0t_1}
        \circ
        \hat{F}^{-1}_{s_1d_0t_1}
        \circ
        \hat{F}_{s_1d_0t_0}\big(
        Y_{s_1d_1,i}(t_0)
        \big) 
        \\- 
        \frac{1}{N_{s_1d_1t_0}}\sum_{i=1}^{N_{s_1d_1t_0}}
        F^{-1}_{s_0d_1t_1}
        \circ
        \hat{F}_{s_0d_1t_0}
        \circ
        \hat{F}^{-1}_{s_0d_0t_0}
        \circ
        \hat{F}_{s_0d_0t_1}
        \circ
        \hat{F}^{-1}_{s_1d_0t_1}
        \circ
        \hat{F}_{s_1d_0t_0}\big(
        Y_{s_1d_1,i}(t_0)
        \big) +\hat{\mu}_1 
        \Big\vert
        \\+
        \sqrt{N}\cdot
        \Big\vert
        \frac{1}{N_{s_1d_1t_0}}\sum_{i=1}^{N_{s_1d_1t_0}}
        F^{-1}_{s_0d_1t_1}
        \circ
        \hat{F}_{s_0d_1t_0}
        \circ
        \hat{F}^{-1}_{s_0d_0t_0}
        \circ
        \hat{F}_{s_0d_0t_1}
        \circ
        \hat{F}^{-1}_{s_1d_0t_1}
        \circ
        \hat{F}_{s_1d_0t_0}\big(
        Y_{s_1d_1,i}(t_0)
        \big) 
        \\- 
        \frac{1}{N_{s_1d_1t_0}}\sum_{i=1}^{N_{s_1d_1t_0}}
        F^{-1}_{s_0d_1t_1}
        \circ
        F_{s_0d_1t_0}
        \circ
        \hat{F}^{-1}_{s_0d_0t_0}
        \circ
        \hat{F}_{s_0d_0t_1}
        \circ
        \hat{F}^{-1}_{s_1d_0t_1}
        \circ
        \hat{F}_{s_1d_0t_0}\big(
        Y_{s_1d_1,i}(t_0)
        \big) -\hat{\mu}_2
        \Big\vert
        \\+
        \sqrt{N}\cdot
        \Big\vert
        \frac{1}{N_{s_1d_1t_0}}\sum_{i=1}^{N_{s_1d_1t_0}}
        F^{-1}_{s_0d_1t_1}
        \circ
        F_{s_0d_1t_0}
        \circ
        \hat{F}^{-1}_{s_0d_0t_0}
        \circ
        \hat{F}_{s_0d_0t_1}
        \circ
        \hat{F}^{-1}_{s_1d_0t_1}
        \circ
        \hat{F}_{s_1d_0t_0}\big(
        Y_{s_1d_1,i}(t_0)
        \big) 
        \\- 
        \frac{1}{N_{s_1d_1t_0}}\sum_{i=1}^{N_{s_1d_1t_0}}
        F^{-1}_{s_0d_1t_1}
        \circ
        F_{s_0d_1t_0}
        \circ
        F^{-1}_{s_0d_0t_0}
        \circ
        \hat{F}_{s_0d_0t_1}
        \circ
        \hat{F}^{-1}_{s_1d_0t_1}
        \circ
        \hat{F}_{s_1d_0t_0}\big(
        Y_{s_1d_1,i}(t_0)
        \big) +\hat{\mu}_3
        \Big\vert
        \\+
        \sqrt{N}\cdot
        \Big\vert
        \frac{1}{N_{s_1d_1t_0}}\sum_{i=1}^{N_{s_1d_1t_0}}
        F^{-1}_{s_0d_1t_1}
        \circ
        F_{s_0d_1t_0}
        \circ
        F^{-1}_{s_0d_0t_0}
        \circ
        \hat{F}_{s_0d_0t_1}
        \circ
        \hat{F}^{-1}_{s_1d_0t_1}
        \circ
        \hat{F}_{s_1d_0t_0}\big(
        Y_{s_1d_1,i}(t_0)
        \big) 
        \\- 
        \frac{1}{N_{s_1d_1t_0}}\sum_{i=1}^{N_{s_1d_1t_0}}
        F^{-1}_{s_0d_1t_1}
        \circ
        F_{s_0d_1t_0}
        \circ
        F^{-1}_{s_0d_0t_0}
        \circ
        F_{s_0d_0t_1}
        \circ
        \hat{F}^{-1}_{s_1d_0t_1}
        \circ
        \hat{F}_{s_1d_0t_0}\big(
        Y_{s_1d_1,i}(t_0)
        \big) -\hat{\mu}_4
        \Big\vert
        \\+
        \sqrt{N}\cdot
        \Big\vert
        \frac{1}{N_{s_1d_1t_0}}\sum_{i=1}^{N_{s_1d_1t_0}}
        F^{-1}_{s_0d_1t_1}
        \circ
        F_{s_0d_1t_0}
        \circ
        F^{-1}_{s_0d_0t_0}
        \circ
        F_{s_0d_0t_1}
        \circ
        \hat{F}^{-1}_{s_1d_0t_1}
        \circ
        \hat{F}_{s_1d_0t_0}\big(
        Y_{s_1d_1,i}(t_0)
        \big) 
        \\- 
        \frac{1}{N_{s_1d_1t_0}}\sum_{i=1}^{N_{s_1d_1t_0}}
        F^{-1}_{s_0d_1t_1}
        \circ
        F_{s_0d_1t_0}
        \circ
        F^{-1}_{s_0d_0t_0}
        \circ
        F_{s_0d_0t_1}
        \circ
        F^{-1}_{s_1d_0t_1}
        \circ
        \hat{F}_{s_1d_0t_0}\big(
        Y_{s_1d_1,i}(t_0)
        \big) +\hat{\mu}_5
        \Big\vert
        \\+
        \sqrt{N}\cdot
        \Big\vert
        \frac{1}{N_{s_1d_1t_0}}\sum_{i=1}^{N_{s_1d_1t_0}}
        F^{-1}_{s_0d_1t_1}
        \circ
        F_{s_0d_1t_0}
        \circ
        F^{-1}_{s_0d_0t_0}
        \circ
        F_{s_0d_0t_1}
        \circ
        F^{-1}_{s_1d_0t_1}
        \circ
        \hat{F}_{s_1d_0t_0}\big(
        Y_{s_1d_1,i}(t_0)
        \big) 
        \\- 
        \frac{1}{N_{s_1d_1t_0}}\sum_{i=1}^{N_{s_1d_1t_0}}
        F^{-1}_{s_0d_1t_1}
        \circ
        F_{s_0d_1t_0}
        \circ
        F^{-1}_{s_0d_0t_0}
        \circ
        F_{s_0d_0t_1}
        \circ
        F^{-1}_{s_1d_0t_1}
        \circ
        F_{s_1d_0t_0}\big(
        Y_{s_1d_1,i}(t_0)
        \big)-\hat{\mu}_6 
        \Big\vert
        \\+
        \sqrt{N}\cdot
        \Big\vert
        \frac{1}{N_{s_1d_1t_0}}\sum_{i=1}^{N_{s_1d_1t_0}}
        F^{-1}_{s_0d_1t_1}
        \circ
        F_{s_0d_1t_0}
        \circ
        F^{-1}_{s_0d_0t_0}
        \circ
        F_{s_0d_0t_1}
        \circ
        F^{-1}_{s_1d_0t_1}
        \circ
        F_{s_1d_0t_0}\big(
        Y_{s_1d_1,i}(t_0)
        \big) 
        \\- 
        \ex{s_1d_1}\big[F^{-1}_{s_0d_1t_1}
        \circ
        F_{s_0d_1t_0}
        \circ
        F^{-1}_{s_0d_0t_0}
        \circ
        F_{s_0d_0t_1}
        \circ
        F^{-1}_{s_1d_0t_1}
        \circ
        F_{s_1d_0t_0}\big(
        Y(t_0)
        \big)\big]-\hat{\upmu}_7
        \Big\vert,
\end{multline}
where the first through sixth terms are $o_p(1)$ due to Lemmas \ref{lem:term1}, \ref{lem:term2}, \ref{lem:term3}, \ref{lem:term4}, \ref{lem:term5} and \ref{lem:term6}, respectively, and the seventh term is $0$ by definition of $\hat{\upmu}_7$ (see Eq.~\ref{eq:muhat7}.)

%% file: appendix/third_term.tex
\begin{lemma}\label{lem:term3}
    Define $\hat{\mu}_3$ as given by Eq.~\eqref{eq:muhat3}. Under \ref{as:estimation},
    \begin{multline}\label{eq:lemterm3}
        \sqrt{N}\cdot
        \Big\vert
        \frac{1}{N_{s_1d_1t_0}}\sum_{i=1}^{N_{s_1d_1t_0}}
        F^{-1}_{s_0d_1t_1}
        \circ
        F_{s_0d_1t_0}
        \circ
        \hat{F}^{-1}_{s_0d_0t_0}
        \circ
        \hat{F}_{s_0d_0t_1}
        \circ
        \hat{F}^{-1}_{s_1d_0t_1}
        \circ
        \hat{F}_{s_1d_0t_0}\big(
        Y_{s_1d_1,i}(t_0)
        \big) 
        \\- 
        \frac{1}{N_{s_1d_1t_0}}\sum_{i=1}^{N_{s_1d_1t_0}}
        F^{-1}_{s_0d_1t_1}
        \circ
        F_{s_0d_1t_0}
        \circ
        F^{-1}_{s_0d_0t_0}
        \circ
        \hat{F}_{s_0d_0t_1}
        \circ
        \hat{F}^{-1}_{s_1d_0t_1}
        \circ
        \hat{F}_{s_1d_0t_0}\big(
        Y_{s_1d_1,i}(t_0)
        \big) 
        +\hat{\mu}_3
        \Big\vert
    \overset{P}{\to}0.\end{multline}
\end{lemma}

\begin{proof}
For ease of notation, define
$\hat{q}_i \coloneqq 
        \hat{F}_{s_0d_0t_1}
        \circ
        \hat{F}^{-1}_{s_1d_0t_1}
        \circ
        \hat{F}_{s_1d_0t_0}\big(
        Y_{s_1d_1,i}(t_0)
        \big) $, and
        $q_i \coloneqq 
        F_{s_0d_0t_1}
        \circ
        F^{-1}_{s_1d_0t_1}
        \circ
        F_{s_1d_0t_0}\big(
        Y_{s_1d_1,i}(t_0)
        \big)$.
Also, define $g(\cdot)\coloneqq F^{-1}_{s_0d_1t_1}
        \circ
        F_{s_0d_1t_0}(\cdot)$.
Using triangle inequality, Eq.~\eqref{eq:lemterm3} can be bounded as
\begin{multline}\label{eq:proofterm3}
    \sqrt{N}\cdot
        \Big\vert
        \frac{1}{N_{s_1d_1t_0}}\sum_{i=1}^{N_{s_1d_1t_0}}
        g
        \circ
        \hat{F}^{-1}_{s_0d_0t_0}
        (\hat{q}_i) 
        - 
        \frac{1}{N_{s_1d_1t_0}}\sum_{i=1}^{N_{s_1d_1t_0}}
        g
        \circ
        F^{-1}_{s_0d_0t_0}
        (\hat{q}_i) 
        +\hat{\mu}_3
        \Big\vert
    \leq\\
    \sqrt{N}\cdot
    \Big\vert
    \frac{1}{N_{s_1d_1t_0}}\sum_{i=1}^{N_{s_1d_1t_0}}
    g
        \circ\hat{F}^{-1}_{s_0d_0t_0}
    \big(\hat{q}_i\big)
    -
    \frac{1}{N_{s_1d_1t_0}}\sum_{i=1}^{N_{s_1d_1t_0}}
    g
        \circ F^{-1}_{s_0d_0t_0}
    \big(\hat{q}_i\big)
    \\+
    \frac{1}{N_{s_1d_1t_0}}\frac{1}{N_{s_0d_0t_0}}\sum_{i=1}^{N_{s_1d_1t_0}}\sum_{j=1}^{N_{s_0d_0t_0}}
    \frac{
    \ind{F_{s_0d_0t_0}\big(Y_{s_0d_0,j}(t_0)\big)\leq
    \hat{q}_i  
    }
    -\hat{q}_i
    }{f_{s_0d_0t_0}\big(
    F^{-1}_{s_0d_0t_0}
    (\hat{q}_i)
    \big)}
    \big(g'\circ f_{s_0d_0t_0}\circ F^{-1}_{s_0d_0t_0}
    (\hat{q}_i)
    \big)
    \Big\vert
    \\
    +
    \sqrt{N}\cdot
    \Big\vert
    \frac{1}{N_{s_1d_1t_0}}\frac{1}{N_{s_0d_0t_0}}\sum_{i=1}^{N_{s_1d_1t_0}}\sum_{j=1}^{N_{s_0d_0t_0}}
    \frac{
    \ind{F_{s_0d_0t_0}\big(Y_{s_0d_0,j}(t_0)\big)\leq
    \hat{q}_i  
    }
    -\hat{q}_i
    }{f_{s_0d_0t_0}\big(
    F^{-1}_{s_0d_0t_0}
    (\hat{q}_i)
    \big)}
    \big(g'\circ f_{s_0d_0t_0}\circ F^{-1}_{s_0d_0t_0}
    (\hat{q}_i)
    \big)
    -\hat{\mu}_3
    \Big\vert
\end{multline}

The first term on the right hand side of \eqref{eq:proofterm3} can be bounded by
\begin{multline}
    \sqrt{N}\cdot
    \frac{1}{N_{s_1d_1t_0}}\cdot\sum_{i=1}^{N_{s_1d_1t_0}}
    \Big\vert
    g\circ\hat{F}^{-1}_{s_0d_0t_0}
    (\hat{q}_i)
    -
    g\circ
    F^{-1}_{s_0d_0t_0}
    (\hat{q}_i)
    \\+
    \frac{1}{N_{s_0d_0t_0}}\sum_{j=1}^{N_{s_0d_0t_0}}
    \frac{
    \ind{F_{s_0d_0t_0}\big(Y_{s_0d_0,j}(t_0)\big)\leq
    \hat{q}_i
    }
    -\hat{q}_i
    }{f_{s_0d_0t_0}\big(
    F^{-1}_{s_0d_0t_0}
    (\hat{q}_i)
    \big)}
    \big(g'\circ f_{s_0d_0t_0}\circ F^{-1}_{s_0d_0t_0}
    (\hat{q}_i)
    \big)
    \Big\vert\leq
    \\
    \sqrt{N}\cdot
    \sup_u
    \Big\vert
    g\circ
    \hat{F}^{-1}_{s_0d_0t_0}
    (u) 
    -
    g\circ
    F^{-1}_{s_0d_0t_0}
    (u) 
    \\+
    \frac{1}{N_{s_0d_0t_0}}\sum_{j=1}^{N_{s_0d_0t_0}}
    \frac{
    \ind{F_{s_0d_0t_0}\big(Y_{s_0d_0,j}(t_0)\big)\leq
    u 
    }
    -u
    }{f_{s_0d_0t_0}\big(
    F^{-1}_{s_0d_0t_0}
    (u) 
    \big)}
    \big(g'\circ f_{s_0d_0t_0}\circ F^{-1}_{s_0d_0t_0}
    (u)
    \big)
    \Big\vert
    \\=
    \sqrt{N}\cdot
    \sup_u
    \Big\vert
    g\circ
    \hat{F}^{-1}_{s_0d_0t_0}
    (u) 
    -
    g\circ
    F^{-1}_{s_0d_0t_0}
    (u) 
    +
    \frac{
    \hat{F}_{s_0d_0t_0}\circ F^{-1}_{s_0d_0t_0}
    (u) 
    -u
    }{f_{s_0d_0t_0}\big(
    F^{-1}_{s_0d_0t_0}
    (u) 
    \big)}
    \big(g'\circ f_{s_0d_0t_0}\circ F^{-1}_{s_0d_0t_0}
    (u)
    \big)
    \Big\vert
\end{multline}
Lemma \ref{lem:gena6} implies that the latter is $o_p(1)$.

The second term on the right hand side of \eqref{eq:proofterm3} can also be rephrased as follows.
\begin{multline}
    \sqrt{N}\cdot
    \Big\vert
    \frac{1}{N_{s_1d_1t_0}}\frac{1}{N_{s_0d_0t_0}}\sum_{i=1}^{N_{s_1d_1t_0}}\sum_{j=1}^{N_{s_0d_0t_0}}
    \frac{
    \ind{F_{s_0d_0t_0}\big(Y_{s_0d_0,j}(t_0)\big)\leq
    \hat{q}_i  
    }
    -\hat{q}_i
    }{f_{s_0d_0t_0}\big(
    F^{-1}_{s_0d_0t_0}
    (\hat{q}_i)
    \big)}
    \big(g'\circ f_{s_0d_0t_0}\circ F^{-1}_{s_0d_0t_0}
    (\hat{q}_i)
    \big)
    -\hat{\mu}_3
    \Big\vert
    \\
    =\sqrt{N}\cdot
    \Big\vert
    \frac{1}{N_{s_1d_1t_0}}\frac{1}{N_{s_0d_0t_0}}\sum_{i=1}^{N_{s_1d_1t_0}}\sum_{j=1}^{N_{s_0d_0t_0}}
    \frac{
    \ind{F_{s_0d_0t_0}\big(Y_{s_0d_0,j}(t_0)\big)\leq
    \hat{q}_i  
    }
    -\hat{q}_i
    }{f_{s_0d_0t_0}\big(
    F^{-1}_{s_0d_0t_0}
    (\hat{q}_i)
    \big)}
    \big(g'\circ f_{s_0d_0t_0}\circ F^{-1}_{s_0d_0t_0}
    (\hat{q}_i)
    \big)
    \\
    -
    \frac{1}{N_{s_1d_1t_0}}\frac{1}{N_{s_0d_0t_0}}\sum_{i=1}^{N_{s_1d_1t_0}}\sum_{j=1}^{N_{s_0d_0t_0}}
    \frac{
    \ind{F_{s_0d_0t_0}\big(Y_{s_0d_0,j}(t_0)\big)\leq
    q_i  
    }
    -q_i
    }{f_{s_0d_0t_0}\big(
    F^{-1}_{s_0d_0t_0}
    (q_i)
    \big)}
    \big(g'\circ f_{s_0d_0t_0}\circ F^{-1}_{s_0d_0t_0}
    (q_i)
    \big)
    \Big\vert
    \\
    =\sqrt{N}\cdot
    \Big\vert
    \frac{1}{N_{s_1d_1t_0}}\sum_{i=1}^{N_{s_1d_1t_0}}
    \frac{
    \hat{F}_{s_0d_0t_0}\circ F^{-1}_{s_0d_0t_0}
    (\hat{q}_i)
    -\hat{q}_i
    }{f_{s_0d_0t_0}\big(
    F^{-1}_{s_0d_0t_0}
    (\hat{q}_i)
    \big)}
    \big(g'\circ f_{s_0d_0t_0}\circ F^{-1}_{s_0d_0t_0}
    (\hat{q}_i)
    \big)
    \\
    -
    \frac{1}{N_{s_1d_1t_0}}\sum_{i=1}^{N_{s_1d_1t_0}}
    \frac{
    \hat{F}_{s_0d_0t_0}\circ F^{-1}_{s_0d_0t_0}
    (q_i)
    -q_i
    }{f_{s_0d_0t_0}\big(
    F^{-1}_{s_0d_0t_0}
    (q_i)
    \big)}
    \big(g'\circ f_{s_0d_0t_0}\circ F^{-1}_{s_0d_0t_0}
    (q_i)
    \big)
    \Big\vert,
\end{multline}
which can be bounded further by
\begin{multline}\label{eq:proof3a1011}
    \sqrt{N}\cdot
    \Big\vert
    \frac{1}{N_{s_1d_1t_0}}\sum_{i=1}^{N_{s_1d_1t_0}}
    \frac{
    \hat{F}_{s_0d_0t_0}\circ F^{-1}_{s_0d_0t_0}(\hat{q}_i)
    -\hat{q}_i
    }{f_{s_0d_0t_0}\big(
    F^{-1}_{s_0d_0t_0}
    (\hat{q}_i ) 
    \big)}
    \big(g'\circ f_{s_0d_0t_0}\circ F^{-1}_{s_0d_0t_0}
    (\hat{q}_i)
    \big)
    \\
    -\frac{1}{N_{s_1d_1t_0}}\sum_{i=1}^{N_{s_1d_1t_0}}
    \frac{
    \hat{F}_{s_0d_0t_0}\circ F^{-1}_{s_0d_0t_0}(q_i)
    -q_i
    }{f_{s_0d_0t_0}\big(
    F^{-1}_{s_0d_0t_0}
    (\hat{q}_i ) 
    \big)}
    \big(g'\circ f_{s_0d_0t_0}\circ F^{-1}_{s_0d_0t_0}
    (\hat{q}_i)
    \big)
    \Big\vert
    \\
    +\sqrt{N}\cdot
    \Big\vert
    \frac{1}{N_{s_1d_1t_0}}\sum_{i=1}^{N_{s_1d_1t_0}}
    \frac{
    \hat{F}_{s_0d_0t_0}\circ F^{-1}_{s_0d_0t_0}(q_i)
    -q_i
    }{f_{s_0d_0t_0}\big(
    F^{-1}_{s_0d_0t_0}
    (\hat{q}_i ) 
    \big)}\big(g'\circ f_{s_0d_0t_0}\circ F^{-1}_{s_0d_0t_0}
    (\hat{q}_i)
    \big)
    \\
    -\frac{1}{N_{s_1d_1t_0}}\sum_{i=1}^{N_{s_1d_1t_0}}
    \frac{
    \hat{F}_{s_0d_0t_0}\circ F^{-1}_{s_0d_0t_0}(q_i)
    -q_i
    }{f_{s_0d_0t_0}\big(
    F^{-1}_{s_0d_0t_0}
    (q_i ) 
    \big)}
    \big(g'\circ f_{s_0d_0t_0}\circ F^{-1}_{s_0d_0t_0}
    (q_i)
    \big)
    \Big\vert
\end{multline}
The first term of Eq.~\eqref{eq:proof3a1011} is bounded by
\begin{multline}\label{eq:proof3a10}
    \sqrt{N}\cdot
    \frac{1}{N_{s_1d_1t_0}}\sum_{i=1}^{N_{s_1d_1t_0}}\Big\vert
    \frac{
    \hat{F}_{s_0d_0t_0}\circ F^{-1}_{s_0d_0t_0}(\hat{q}_i)
    -\hat{q}_i
    }{f_{s_0d_0t_0}\big(
    F^{-1}_{s_0d_0t_0}
    (\hat{q}_i ) 
    \big)}
    -
    \frac{
    \hat{F}_{s_0d_0t_0}\circ F^{-1}_{s_0d_0t_0}(q_i)
    -q_i
    }{f_{s_0d_0t_0}\big(
    F^{-1}_{s_0d_0t_0}
    (\hat{q}_i ) 
    \big)}
    \Big\vert\cdot
    \big\vert g'\circ f_{s_0d_0t_0}\circ F^{-1}_{s_0d_0t_0}
    (\hat{q}_i)
    \big\vert
    \leq\\
    \sqrt{N}\cdot\sup_u
    \vert
    \frac{g'\circ f_{s_0d_0t_0}\circ F^{-1}_{s_0d_0t_0}
    (u)}{f_{s_0d_0t_0}\big(
    F^{-1}_{s_0d_0t_0}
    (u) 
    \big)}
    \vert
    \cdot
    \sup_y\Big\vert
    \hat{F}_{s_0d_0t_0}\circ F^{-1}_{s_0d_0t_0}\circ
    \hat{F}_{s_0d_0t_0}
        \circ
        \hat{F}^{-1}_{s_0d_1t_0}
        \circ
        \hat{F}_{s_0d_1t_1}
        (y)
    -
    \hat{F}_{s_0d_0t_0}
        \circ
        \hat{F}^{-1}_{s_0d_1t_0}
        \circ
        \hat{F}_{s_0d_1t_1}
        (y)
    \\
    -\big(
    \hat{F}_{s_0d_0t_0}\circ F^{-1}_{s_0d_0t_0}\circ
    F_{s_0d_0t_0}
        \circ
        F^{-1}_{s_0d_1t_0}
        \circ
        F_{s_0d_1t_1}
        (y)
    -
    F_{s_0d_0t_0}
        \circ
        F^{-1}_{s_0d_1t_0}
        \circ
        F_{s_0d_1t_1}
        (y)
    \big)
    \Big\vert
    =\\
    \sqrt{N}\cdot\sup_u
    \vert
    \frac{g'\circ f_{s_0d_0t_0}\circ F^{-1}_{s_0d_0t_0}
    (u)}{f_{s_0d_0t_0}\big(
    F^{-1}_{s_0d_0t_0}
    (u) 
    \big)}
    \vert
    \cdot
    \sup_y\Big\vert
    \hat{F}_{s_0d_0t_0}\circ F^{-1}_{s_0d_0t_0}\circ
    \hat{F}_{s_0d_0t_0}
        \circ
        \hat{F}^{-1}_{s_0d_1t_0}
        \circ
        \hat{F}_{s_0d_1t_1}
        (y)
    -\\
    \hat{F}_{s_0d_0t_0}\circ F^{-1}_{s_0d_0t_0}\circ
    F_{s_0d_0t_0}
        \circ
        F^{-1}_{s_0d_1t_0}
        \circ
        F_{s_0d_1t_1}
        (y)
    \\
    -\big(
    F_{s_0d_0t_0}\circ F^{-1}_{s_0d_0t_0}\circ
    \hat{F}_{s_0d_0t_0}
        \circ
        \hat{F}^{-1}_{s_0d_1t_0}
        \circ
        \hat{F}_{s_0d_1t_1}
        (y)
    -
    F_{s_0d_0t_0}\circ F^{-1}_{s_0d_0t_0}\circ
    F_{s_0d_0t_0}
        \circ
        F^{-1}_{s_0d_1t_0}
        \circ
        F_{s_0d_1t_1}
        (y)
    \big)
    \Big\vert
    =\\
    \sqrt{N}\cdot\sup_u
    \vert
    \frac{g'\circ f_{s_0d_0t_0}\circ F^{-1}_{s_0d_0t_0}
    (u)}{f_{s_0d_0t_0}\big(
    F^{-1}_{s_0d_0t_0}
    (u) 
    \big)}
    \vert
    \cdot
    \sup_y\Big\vert
    \hat{F}_{s_0d_0t_0}\circ F^{-1}_{s_0d_0t_0}\circ
    \hat{m}
        (y)
    -
    \hat{F}_{s_0d_0t_0}\circ F^{-1}_{s_0d_0t_0}\circ
    m
        (y)
    \\
    -\big(
    F_{s_0d_0t_0}\circ F^{-1}_{s_0d_0t_0}\circ
    \hat{m}
        (y)
    -
    F_{s_0d_0t_0}\circ F^{-1}_{s_0d_0t_0}\circ
    m
        (y)
    \big)
    \Big\vert,
\end{multline}
where $m(\cdot)=F_{s_0d_0t_1}
        \circ
        F^{-1}_{s_1d_0t_1}
        \circ
        F_{s_1d_0t_0}
        (\cdot)$, and $\hat{m}(\cdot)=
        \hat{F}_{s_0d_0t_1}
        \circ
        \hat{F}^{-1}_{s_1d_0t_1}
        \circ
        \hat{F}_{s_1d_0t_0}
        (\cdot)$.
The term $\sup_u
    \vert
    \frac{g'\circ f_{s_0d_0t_0}\circ F^{-1}_{s_0d_0t_0}
    (u)}{f_{s_0d_0t_0}\big(
    F^{-1}_{s_0d_0t_0}
    (u) 
    \big)}
    \vert$
is bounded due to \ref{as:estimation}.
Choose $\delta=1/3$.
Due to Lemma \ref{lem:genuniform}, $\hat{m}
        (y)-
        m
        (y)=o_p(N^{-\delta})$ and $F^{-1}_{s_0d_0t_0}$ has the required smoothness properties due to \ref{as:estimation}, $F^{-1}_{s_0d_0t_0}\circ\hat{m}(y)-F^{-1}_{s_0d_0t_0}\circ m(y)=o_p(N^{-\delta})$.
Therefore, Lemma A.5 of \citet{athey2006identification} with the choice of $\eta=1/2$ implies that the right hand side of Eq.~\eqref{eq:proof3a10} is $o_p(1)$.

Finally, the second term of Eq.~\eqref{eq:proof3a1011} is upper-bounded by
\begin{multline}
    \sqrt{N}\cdot
    \frac{1}{N_{s_1d_1t_0}}\sum_{i=1}^{N_{s_1d_1t_0}}
    \Big\vert
    \hat{F}_{s_0d_0t_0}\circ F^{-1}_{s_0d_0t_0}(q_i)
    -q_i
    \Big\vert\cdot\Big\vert\frac{
    g'\circ f_{s_0d_0t_0}\circ F^{-1}_{s_0d_0t_0}
    (\hat{q}_i)
    }{f_{s_0d_0t_0}\big(
    F^{-1}_{s_0d_0t_0}
    (\hat{q}_i)
    \big)}
    -
    \frac{g'\circ f_{s_0d_0t_0}\circ F^{-1}_{s_0d_0t_0}
    (q_i)}{f_{s_0d_0t_0}\big(
    F^{-1}_{s_0d_0t_0}
    (q_i)
    \big)}
    \Big\vert\\
    \leq
    \Big(N^{1/4}
    \sup_y
    \Big\vert
    \hat{F}_{s_0d_0t_0}\circ F^{-1}_{s_0d_0t_0}\big(m(
        y)
        \big)
    -m(
        y
        )
    \Big\vert\Big)\\\times
    \Big(N^{1/4}
    \sup_y\Big\vert
    \frac{
    g'\circ f_{s_0d_0t_0}\circ F^{-1}_{s_0d_0t_0}
    \circ\hat{m}(y)}{f_{s_0d_0t_0}\circ
    F^{-1}_{s_0d_0t_0}
    \circ\hat{m}(
        y)}
    -
    \frac{
    g'\circ f_{s_0d_0t_0}\circ F^{-1}_{s_0d_0t_0}
    \circ m(y)
    }{f_{s_0d_0t_0}\circ
    F^{-1}_{s_0d_0t_0}
    \circ m(
        y)
    }
    \Big\vert\Big)
    \\
    =
    \Big(N^{1/4}
    \sup_y
    \Big\vert
    \hat{F}_{s_0d_0t_0}\circ F^{-1}_{s_0d_0t_0}\big(m(
        y)
        \big)
    -F_{s_0d_0t_0}\circ F^{-1}_{s_0d_0t_0}\big(m(
        y
        )\big)
    \Big\vert\Big)\\\times\Big(N^{1/4}
    \sup_y\Big\vert
    \frac{
    g'\circ f_{s_0d_0t_0}\circ F^{-1}_{s_0d_0t_0}
    \circ\hat{m}(y)}{f_{s_0d_0t_0}\circ
    F^{-1}_{s_0d_0t_0}
    \circ\hat{m}(
        y)}
    -
    \frac{
    g'\circ f_{s_0d_0t_0}\circ F^{-1}_{s_0d_0t_0}
    \circ m(y)
    }{f_{s_0d_0t_0}\circ
    F^{-1}_{s_0d_0t_0}
    \circ m(
        y)
    }
    \Big\vert\Big)
    \\
    =
    \Big(N^{1/4}
    \sup_y
    \Big\vert
    \hat{F}_{s_0d_0t_0}(
        y)
    -F_{s_0d_0t_0}(
        y
        )
    \Big\vert\Big)\\\times\Big(N^{1/4}
    \sup_y\Big\vert
    \frac{
    g'\circ f_{s_0d_0t_0}\circ F^{-1}_{s_0d_0t_0}
    \circ\hat{m}(y)}{f_{s_0d_0t_0}\circ
    F^{-1}_{s_0d_0t_0}
    \circ\hat{m}(
        y)}
    -
    \frac{
    g'\circ f_{s_0d_0t_0}\circ F^{-1}_{s_0d_0t_0}
    \circ m(y)
    }{f_{s_0d_0t_0}\circ
    F^{-1}_{s_0d_0t_0}
    \circ m(
        y)
    }
    \Big\vert\Big),
\end{multline}
where the first term is $o_p(1)$ due to Lemma \ref{lem:uniform}, and the second term is $o_p(1)$ due to Lemma \ref{lem:genuniform} and Assumption \ref{as:estimation}.
\end{proof}

%% file: appendix/fourth_term.tex
\begin{lemma}\label{lem:term2}
    Define $\hat{\mu}_2$ as given by Eq.~\eqref{eq:muhat2}. Under \ref{as:estimation},
    \begin{multline}\label{eq:lemterm2}
        \sqrt{N}\cdot
        \Big\vert
        \frac{1}{N_{s_1d_1t_0}}\sum_{i=1}^{N_{s_1d_1t_0}}
        F^{-1}_{s_0d_1t_1}
        \circ
        \hat{F}_{s_0d_1t_0}
        \circ
        \hat{F}^{-1}_{s_0d_0t_0}
        \circ
        \hat{F}_{s_0d_0t_1}
        \circ
        \hat{F}^{-1}_{s_1d_0t_1}
        \circ
        \hat{F}_{s_1d_0t_0}\big(
        Y_{s_1d_1,i}(t_0)
        \big) 
        \\- 
        \frac{1}{N_{s_1d_1t_0}}\sum_{i=1}^{N_{s_1d_1t_0}}
        F^{-1}_{s_0d_1t_1}
        \circ
        F_{s_0d_1t_0}
        \circ
        \hat{F}^{-1}_{s_0d_0t_0}
        \circ
        \hat{F}_{s_0d_0t_1}
        \circ
        \hat{F}^{-1}_{s_1d_0t_1}
        \circ
        \hat{F}_{s_1d_0t_0}\big(
        Y_{s_1d_1,i}(t_0)
        \big) 
        -\hat{\mu}_2
        \Big\vert
    \overset{P}{\to}0.\end{multline}
\end{lemma}
\begin{proof}
For ease of notation, define 
$\hat{q}_i \coloneqq \hat{F}^{-1}_{s_0d_0t_0}
        \circ
        \hat{F}_{s_0d_0t_1}
        \circ
        \hat{F}^{-1}_{s_1d_0t_1}
        \circ
        \hat{F}_{s_1d_0t_0}\big(
        Y_{s_1d_1,i}(t_0)
        \big) $, and
        $q_i \coloneqq F^{-1}_{s_0d_0t_0}
        \circ
        F_{s_0d_0t_1}
        \circ
        F^{-1}_{s_1d_0t_1}
        \circ
        F_{s_1d_0t_0}\big(
        Y_{s_1d_1,i}(t_0)
        \big)$.
Also define $g(\cdot)\coloneqq F^{-1}_{s_0d_1t_1}(\cdot)$.
Using triangle inequality and plugging in $\hat{\mu}_2$, Eq.~\eqref{eq:lemterm2} can be bounded as
\begin{multline}\label{eq:prooflema3}
    \sqrt{N}\cdot
        \Big\vert
        \frac{1}{N_{s_1d_1t_0}}\sum_{i=1}^{N_{s_1d_1t_0}}
        g
        \circ
        \hat{F}_{s_0d_1t_0}
         (\hat{q}_i)
        - 
        \frac{1}{N_{s_1d_1t_0}}\sum_{i=1}^{N_{s_1d_1t_0}}
        g
        \circ
        F_{s_0d_1t_0}(\hat{q}_i)
        -\hat{\mu}_2
        \Big\vert
        \\\leq 
        \sqrt{N}\cdot
        \Big\vert
        \frac{1}{N_{s_1d_1t_0}}\sum_{i=1}^{N_{s_1d_1t_0}}\big(
        g
        \circ
        \hat{F}_{s_0d_1t_0}
         (\hat{q}_i)
         -g
        \circ
        F_{s_0d_1t_0}(\hat{q}_i)\big)
        \\- 
        \frac{1}{N_{s_1d_1t_0}}\frac{1}{N_{s_0d_1t_0}}\sum_{i=1}^{N_{s_1d_1t_0}}\sum_{j=1}^{N_{s_0d_1t_0}}
        \big(
        \ind{
        Y_{s_0d_1,j}(t_0)\leq \hat{q}_i
        }
        -F_{s_0d_1t_0}(\hat{q}_i)
        \big)
        \cdot g' \circ F_{s_0d_1t_0}(\hat{q}_i)
        \Big\vert
        \\+ 
        \sqrt{N}\cdot
        \Big\vert
        \frac{1}{N_{s_1d_1t_0}}\frac{1}{N_{s_0d_1t_0}}\sum_{i=1}^{N_{s_1d_1t_0}}\sum_{j=1}^{N_{s_0d_1t_0}}
        \big(
        \ind{
        Y_{s_0d_1,j}(t_0)\leq \hat{q}_i
        }
        -F_{s_0d_1t_0}(\hat{q}_i)
        \big)
        \cdot g' \circ F_{s_0d_1t_0}(\hat{q}_i)
        \\- 
        \frac{1}{N_{s_1d_1t_0}}\frac{1}{N_{s_0d_1t_0}}\sum_{i=1}^{N_{s_1d_1t_0}}\sum_{j=1}^{N_{s_0d_1t_0}}
        \big(
        \ind{
        Y_{s_0d_1,j}(t_0)\leq q_i
        }
        -F_{s_0d_1t_0}(q_i)
        \big)\cdot g' \circ F_{s_0d_1t_0}(q_i)
        \Big\vert
\end{multline}

The first term on the right hand side of Eq.~\eqref{eq:prooflema3} can be bounded by
\begin{multline}\label{eq:prooflema3first}
    \sqrt{N}\cdot
        \sup_y\Big\vert\big(
        g
        \circ
        \hat{F}_{s_0d_1t_0}
         (y)
         -g
        \circ
        F_{s_0d_1t_0}(y)\big)
        - 
        \frac{1}{N_{s_0d_1t_0}}\sum_{j=1}^{N_{s_0d_1t_0}}
        \big(
        \ind{
        Y_{s_0d_1,j}(t_0)\leq y
        }
        -F_{s_0d_1t_0}(y)
        \big)\cdot g' \circ F_{s_0d_1t_0}(y)
        \Big\vert
        \\= 
    \sqrt{N}\cdot
        \sup_y\Big\vert
        g
        \circ
        \hat{F}_{s_0d_1t_0}
         (y)
         -g
        \circ
        F_{s_0d_1t_0}(y)
        - 
        \big(
        \hat{F}_{s_0d_1t_0}(y)
        -F_{s_0d_1t_0}(y)
        \big)\cdot g' \circ F_{s_0d_1t_0}(y)
        \Big\vert
        \\\overset{(a)}{\leq} 
    \sqrt{N}\cdot
    \sup_u\big\vert
    g''(u)
    \big\vert
    \cdot
        \sup_y\big\vert\hat{F}_{s_0d_1t_0}(y)
        -F_{s_0d_1t_0}(y)\big\vert^2
        \\= 
    \sqrt{N}\cdot
    \sup_u\big\vert
    g''(u)
    \big\vert
    \cdot
        \sup_y\big\vert\hat{F}_{s_0d_1t_0}(y)
        -F_{s_0d_1t_0}(y)\big\vert^2,
\end{multline}
where $g''(u) = -\frac{1}{f_{s_0d_1t_1}(y)^3}\frac{\partial f_{s_0d_1t_1}}{\partial y}(u)$.
Note that in $(a)$ we used the expansion of $g\equiv F^{-1}_{s_0d_1t_1}$ around $F_{s_0d_1t_0}(y)$.
The right hand side of Eq.~\eqref{eq:prooflema3first} is $o_p(1)$ due to Assumption \ref{as:estimation}, and Lemma \ref{lem:uniform}.

Now consider the second term of Eq.~\eqref{eq:prooflema3}, which is equal to
\begin{multline}
    \sqrt{N}\cdot
        \Big\vert
        \frac{1}{N_{s_1d_1t_0}}\sum_{i=1}^{N_{s_1d_1t_0}}
        \big(
        \hat{F}_{s_0d_1t_0}(\hat{q}_i)
        -F_{s_0d_1t_0}(\hat{q}_i)
        \big)\cdot g' \circ F_{s_0d_1t_0}(\hat{q}_i)
        \\- 
        \frac{1}{N_{s_1d_1t_0}}\sum_{i=1}^{N_{s_1d_1t_0}}
        \big(
        \hat{F}_{s_0d_1t_0}(q_i)
        -F_{s_0d_1t_0}(q_i)
        \big)\cdot g' \circ F_{s_0d_1t_0}(q_i)
        \Big\vert.
\end{multline}
The latter can be bounded using triangle inequality by:
\begin{multline}\label{eq:prooflema3second}
    \sqrt{N}\Big\vert
        \frac{1}{N_{s_1d_1t_0}}\sum_{i=1}^{N_{s_1d_1t_0}}
        \big(
        \hat{F}_{s_0d_1t_0}(\hat{q}_i)
        -F_{s_0d_1t_0}(\hat{q}_i)
        \big)\cdot g' \circ F_{s_0d_1t_0}(\hat{q}_i)
        \\- 
        \frac{1}{N_{s_1d_1t_0}}\sum_{i=1}^{N_{s_1d_1t_0}}
        \big(
        \hat{F}_{s_0d_1t_0}(q_i)
        -F_{s_0d_1t_0}(q_i)
        \big)\cdot g' \circ F_{s_0d_1t_0}(\hat{q}_i)
        \Big\vert
        \\+
        \sqrt{N}\Big\vert
        \frac{1}{N_{s_1d_1t_0}}\sum_{i=1}^{N_{s_1d_1t_0}}
        \big(
        \hat{F}_{s_0d_1t_0}(q_i)
        -F_{s_0d_1t_0}(q_i)
        \big)\cdot g' \circ F_{s_0d_1t_0}(\hat{q}_i)
        -\\ 
        \frac{1}{N_{s_1d_1t_0}}\sum_{i=1}^{N_{s_1d_1t_0}}
        \big(
        \hat{F}_{s_0d_1t_0}(q_i)
        -F_{s_0d_1t_0}(q_i)
        \big)\cdot g' \circ F_{s_0d_1t_0}(q_i)
        \Big\vert.
\end{multline}
The first term on the right hand side of Eq.~\eqref{eq:prooflema3second} is bounded by
\begin{equation}
    \sqrt{N}\cdot
    \sup_y\big\vert g' \circ F_{s_0d_1t_0}(y)\big\vert
    \cdot
    \sup_y\big\vert
    \hat{F}_{s_0d_1t_0}\circ \hat{m}(y)
        -F_{s_0d_1t_0}
        \circ m(y)
    -\big(
    \hat{F}_{s_0d_1t_0}
    \circ m(y)
        -F_{s_0d_1t_0}
        \circ m(y)
    \big)
    \big\vert,
\end{equation}
where $m(y)=F^{-1}_{s_0d_0t_0}
        \circ
        F_{s_0d_0t_1}
        \circ
        F^{-1}_{s_1d_0t_1}
        \circ
        F_{s_1d_0t_0}(y)$, and $\hat{m}(y)=\hat{F}^{-1}_{s_0d_0t_0}
        \circ
        \hat{F}_{s_0d_0t_1}
        \circ
        \hat{F}^{-1}_{s_1d_0t_1}
        \circ
        \hat{F}_{s_1d_0t_0}(y)$.
The latter is $o_p(1)$ due to Assumption \ref{as:estimation} and Lemma A.5 of \citet{athey2006identification} with the choice of $\eta=1/2,\delta=1/3$.
Note that $
    \hat{m}(y)
    -
    m(y)=o_p(N^{-\delta})
    $ due to Lemma \ref{lem:genuniform} for any $\delta<1/2$, and the conditions of Lemma A.5 of \cite{athey2006identification} are satisfied.

Finally, the second term of Eq.~\eqref{eq:prooflema3second} is bounded by
\begin{multline}
    \sqrt{N}\cdot
    \sup_y\big\vert\hat{F}_{s_0d_1t_0}(y)-F_{s_0d_1t_0}(y)\big\vert\cdot
    \max_i
    \Big\vert
        g'\circ F_{s_0d_1t_0}(\hat{q}_i)-
        g'\circ F_{s_0d_1t_0}(q_i)
    \Big\vert
    \\
    \leq
    \Big(N^{1/4}\sup_y\big\vert\hat{F}_{s_0d_1t_0}(y)-F_{s_0d_1t_0}(y)\big\vert\Big)
    \cdot
    \Big(N^{1/4}
    \sup_y
    \big\vert
    g'\circ F_{s_0d_1t_0}\circ\hat{m}(y)-
        g'\circ F_{s_0d_1t_0}\circ m(y)
    \big\vert\Big),
\end{multline}
where the first term is $o_p(1)$ due to Lemma \ref{lem:uniform}, and second term is $o_p(1)$ due to Assumption \ref{as:estimation} and the fact that 
$\hat{m}(y)
    -
    m(y)=o_p(N^{-1/4})
    $ as above.
\end{proof}

\begin{lemma}\label{lem:term4}
    Define $\hat{\mu}_4$ as given by Eq.~\eqref{eq:muhat4}. Under \ref{as:estimation},
    \begin{multline}
        \sqrt{N}\cdot
        \Big\vert
        \frac{1}{N_{s_1d_1t_0}}\sum_{i=1}^{N_{s_1d_1t_0}}
        F^{-1}_{s_0d_1t_1}
        \circ
        F_{s_0d_1t_0}
        \circ
        F^{-1}_{s_0d_0t_0}
        \circ
        \hat{F}_{s_0d_0t_1}
        \circ
        \hat{F}^{-1}_{s_0d_1t_1}
        \circ
        \hat{F}_{s_0d_1t_0}\big(
        Y_{s_1d_1,i}(t_0)
        \big) 
        \\- 
        \frac{1}{N_{s_1d_1t_0}}\sum_{i=1}^{N_{s_1d_1t_0}}
        F^{-1}_{s_0d_1t_1}
        \circ
        F_{s_0d_1t_0}
        \circ
        F^{-1}_{s_0d_0t_0}
        \circ
        F_{s_0d_0t_1}
        \circ
        \hat{F}^{-1}_{s_0d_1t_1}
        \circ
        \hat{F}_{s_0d_1t_0}\big(
        Y_{s_1d_1,i}(t_0)
        \big) 
        -\hat{\mu}_4
        \Big\vert
    \overset{P}{\to}0.\end{multline}
\end{lemma}
\begin{proof}
    The proof is analogous to that of Lemma \ref{lem:term2}.
    To adapt the proof, we just need to replace the following definitions:
    Define $\hat{q}_i\coloneqq 
        \hat{F}^{-1}_{s_1d_0t_1}
        \circ
        \hat{F}_{s_1d_0t_0}\big(
        Y_{s_1d_1,i}(t_0)
        \big) $, 
        $q_i \coloneqq 
        F^{-1}_{s_1d_0t_1}
        \circ
        F_{s_1d_0t_0}\big(
        Y_{s_1d_1,i}(t_0)
        \big)$,
and $g(\cdot)\coloneqq F^{-1}_{s_0d_1t_1}\circ F_{s_0d_1t_0}\circ F^{-1}_{s_0d_0t_0}(\cdot)$.
The rest of the proof is identical.
\end{proof}

\begin{lemma}\label{lem:term6}
    Define $\hat{\mu}_6$ as given by Eq.~\eqref{eq:muhat6}. Under \ref{as:estimation},
    \begin{multline}\label{eq:lemterm6}
        \sqrt{N}\cdot
        \Big\vert
        \frac{1}{N_{s_1d_1t_0}}\sum_{i=1}^{N_{s_1d_1t_0}}
        F^{-1}_{s_0d_1t_1}
        \circ
        F_{s_0d_1t_0}
        \circ
        F^{-1}_{s_0d_0t_0}
        \circ
        F_{s_0d_0t_1}
        \circ
        F^{-1}_{s_1d_0t_1}
        \circ
        \hat{F}_{s_1d_0t_0}\big(
        Y_{s_1d_1,i}(t_0)
        \big) 
        \\- 
        \frac{1}{N_{s_1d_1t_0}}\sum_{i=1}^{N_{s_1d_1t_0}}
        F^{-1}_{s_0d_1t_1}
        \circ
        F_{s_0d_1t_0}
        \circ
        F^{-1}_{s_0d_0t_0}
        \circ
        F_{s_0d_0t_1}
        \circ
        F^{-1}_{s_1d_0t_1}
        \circ
        F_{s_1d_0t_0}\big(
        Y_{s_1d_1,i}(t_0)
        \big) 
        -\hat{\mu}_6
        \Big\vert
    \overset{P}{\to}0.\end{multline}
\end{lemma}
\begin{proof}
For ease of notation, define $y_i\coloneqq Y_{s_1d_1,i}(t_0)$, and $g(\cdot)\coloneqq F^{-1}_{s_0d_1t_1}
        \circ
        F_{s_0d_1t_0}
        \circ
        F^{-1}_{s_0d_0t_0}
        \circ
        F_{s_0d_0t_1}
        \circ
        F^{-1}_{s_1d_0t_1}(\cdot)$.
Using triangle inequality and plugging in $\hat{\mu}_6$, Eq.~\eqref{eq:lemterm6} can be bounded as
\begin{multline}\label{eq:prooflemterm6}
    \sqrt{N}\cdot
        \Big\vert
        \frac{1}{N_{s_1d_1t_0}}\sum_{i=1}^{N_{s_1d_1t_0}}
        g
        \circ
        \hat{F}_{s_1d_0t_0}
         (y_i)
        - 
        \frac{1}{N_{s_1d_1t_0}}\sum_{i=1}^{N_{s_1d_1t_0}}
        g
        \circ
        F_{s_1d_0t_0}(y_i)
        -\hat{\mu}_6
        \Big\vert
        \\= 
        \sqrt{N}\cdot
        \Big\vert
        \frac{1}{N_{s_1d_1t_0}}\sum_{i=1}^{N_{s_1d_1t_0}}\big(
        g
        \circ
        \hat{F}_{s_1d_0t_0}
         (y_i)
         -g
        \circ
        F_{s_1d_0t_0}(y_i)
        \\- 
        \frac{1}{N_{s_1d_0t_0}}\sum_{j=1}^{N_{s_1d_0t_0}}
        \big(
        \ind{
        Y_{s_1d_0,j}(t_0)\leq y_i
        }
        -F_{s_1d_0t_0}(y_i)
        \big)\cdot g' \circ F_{s_1d_0t_0}(y_i)
        \big)
        \Big\vert
        \\\leq
    \sqrt{N}\cdot
        \sup_y\Big\vert
        g
        \circ
        \hat{F}_{s_1d_0t_0}
         (y)
         -g
        \circ
        F_{s_1d_0t_0}(y)
        - 
        \frac{1}{N_{s_1d_0t_0}}\sum_{j=1}^{N_{s_1d_0t_0}}
        \big(
        \ind{
        Y_{s_1d_0,j}(t_0)\leq y
        }
        -F_{s_1d_0t_0}(y)
        \big)\cdot g' \circ F_{s_1d_0t_0}(y)
        \Big\vert
        \\= 
    \sqrt{N}\cdot
        \sup_y\Big\vert
        g
        \circ
        \hat{F}_{s_1d_0t_0}
         (y)
         -g
        \circ
        F_{s_1d_0t_0}(y)
        - 
        \big(
        \hat{F}_{s_1d_0t_0}(y)
        -F_{s_1d_0t_0}(y)
        \big)\cdot g' \circ F_{s_1d_0t_0}(y)
        \Big\vert
        \\\overset{(a)}{\leq} 
    \sqrt{N}\cdot
    \sup_u\big\vert
    g''(u)
    \big\vert
    \cdot
        \sup_y\big\vert\hat{F}_{s_1d_0t_0}(y)
        -F_{s_1d_0t_0}(y)\big\vert^2
        \\= 
    \sqrt{N}\cdot
    \sup_u\big\vert
    g''(u)
    \big\vert
    \cdot
        \sup_y\big\vert\hat{F}_{s_1d_0t_0}(y)
        -F_{s_1d_0t_0}(y)\big\vert^2,
\end{multline}
where in $(a)$ we used the expansion of $g$ around $F_{s_1d_0t_0}(y)$.
The right hand side of Eq.~\eqref{eq:prooflemterm6} is $o_p(1)$ due to Assumption \ref{as:estimation}, and Lemma \ref{lem:uniform}.
\end{proof}